\documentclass[11pt]{article}
\usepackage[utf8]{inputenc}
\usepackage[english]{babel}
\usepackage[T1]{fontenc}
\usepackage[margin=1.15in]{geometry}
\usepackage{amsmath}
\usepackage{amsthm}
\usepackage{amsfonts}
\usepackage{graphicx,subfig}
\usepackage{float}
\usepackage{pgfplots}
\pgfplotsset{compat=1.14, set layers}
\usepackage{bm}

\usepackage{mathtools}
\usepackage{footmisc}
\usepackage{amssymb}
\usepackage{subfig}
\usepackage{enumitem}
\usepackage{setspace}
\usepackage{caption} 
\usepackage{float}
\usepackage{thm-restate}
\usepackage{xspace,xcolor}
\usepackage{thmtools}
\usepackage{framed}
\usepackage{soul}
\usepackage{algorithm}
\usepackage[noend]{algpseudocode}

\definecolor{refcolor}{rgb}{0.23, 0.27, 0.29}
\usepackage[breaklinks,colorlinks,urlcolor={refcolor},citecolor={refcolor},linkcolor={refcolor}]{hyperref}
\usepackage[capitalize, nameinlink]{cleveref}

\newtheorem{theorem}{Theorem}[section]
\newtheorem*{theorem*}{Theorem}
\newtheorem{lemma}[theorem]{Lemma}
\newtheorem{fact}[theorem]{Fact}

\newtheorem{corollary}[theorem]{Corollary}
\newtheorem{definition}[theorem]{Definition}
\newtheorem{remark}[theorem]{Remark}

\newtheorem{observation}[theorem]{Observation}

\newtheorem*{fact*}{Fact}

\newtheorem*{lemma*}{Lemma}
\newtheorem*{definition*}{Definition}

\Crefname{paragraph}{Paragraph}{Paragraphs}
\Crefname{observation}{Observation}{Observations}
\Crefname{restatement}{Restatement}{Restatements}
\Crefname{fact}{Fact}{Facts}

\newcommand{\ampc}{\textsf{AMPC}\xspace}
\newcommand{\mpc}{\textsf{MPC}\xspace}
\newcommand{\lca}{\textsf{LCA}\xspace}
\newcommand{\local}{\textsf{LOCAL}\xspace}
\newcommand{\volume}{\textsf{VOLUME}\xspace}
\newcommand{\clique}{\textsf{CONGESTED CLIQUE}\xspace}

\DeclareMathOperator*{\E}{\mathbf{E}}

\newcommand{\eps}{\varepsilon}
\newcommand{\poly}{\operatorname{\text{{\rm poly}}}}

\newcommand{\arr}{\xrightarrow{}}

\newcommand{\larr}{\xleftarrow{}}

\begin{document}
	
\begin{center}
	\begin{minipage}[H]{14.5cm} 
		
		\begin{center}
		{\huge \bf Adaptive Massively Parallel Coloring in Sparse Graphs}
			\end{center}
		\vspace{1cm}
		
		{\large \textbf{Rustam Latypov\footnotemark[1]}, Aalto University -- \href{mailto:rustam.latypov@aalto.fi}{\texttt{rustam.latypov@aalto.fi}}} \vspace{1mm}\\
		{\large \textbf{Yannic Maus\footnotemark[2]}, TU Graz -- \href{mailto:yannic.maus@ist.tugraz.at}{\texttt{yannic.maus@ist.tugraz.at}}} \vspace{1mm}\\
		{\large \textbf{Shreyas Pai\footnotemark[3]}, IIT Madras -- \href{mailto:shreyaspai.24@gmail.com}{\texttt{shreyas@cse.iitm.ac.in}}} \vspace{1mm}\\
		{\large \textbf{Jara Uitto}, Aalto University -- \href{mailto:jara.uitto@aalto.fi}{\texttt{jara.uitto@aalto.fi}}} \vspace{1mm}\\
		
		\vspace{5mm}
		
		\begin{abstract}
			Classic symmetry-breaking problems on graphs have gained a lot of attention in models of modern parallel computation. The Adaptive Massively Parallel Computation (\ampc) is a model that captures the central challenges in data center computations. Chang et al.\ [PODC'2019] gave an extremely fast, constant time, algorithm for the $(\Delta + 1)$-coloring problem, where $\Delta$ is the maximum degree of an input graph of $n$ nodes. The algorithm works in the most restrictive \emph{low-space} setting, where each machine has $n^{\delta}$ local space for a constant $0 < \delta < 1$.
			
			The standard approaches for $(\Delta + 1)$-coloring are ignorant about the graph topology in the following sense: They exploit the property that \emph{any} partial coloring can be extended to a feasible $(\Delta + 1)$-coloring of the whole graph. For most graphs, the chromatic number is much smaller than $\Delta + 1$ and we would like to find colorings with fewer colors. However, as soon as we have fewer than $\Delta + 1$ colors, it might not be possible to complete partial colorings.
			
			In this work, we study the vertex-coloring problem in sparse graphs parameterized by their \emph{arboricity} $\alpha$, a standard measure for sparsity. We give deterministic algorithms that in constant, or almost constant, time give $\poly \alpha$ and $O(\alpha)$-colorings, where $\alpha$ can be arbitrarily smaller than $\Delta$. A strong and standard approach to compute arboricity-dependent colorings is through the Nash-Williams forest decomposition, which gives rise to an (acyclic) orientation of the edges such that each node has a small out-degree.
			
			Our main technical contribution is giving efficient deterministic algorithms to compute these orientations and showing how to leverage them to find colorings in low-space \ampc. A key technical challenge is that the color of a node may depend on \emph{almost all} of the other nodes in the graph and these dependencies cannot be stored on a single machine. Nevertheless, our novel and careful exploration technique yields the orientation, and the arboricity-dependent coloring, with a sublinear number of adaptive queries per node.
		\end{abstract}
		
		\end{minipage}
	
	\thispagestyle{empty}
	\footnotetext[1]{Supported by the Research Council of Finland, Grant 334238.}
	\footnotetext[2]{Supported by the Austrian Science Fund (FWF), Grant P36280-N.}
	\footnotetext[3]{This work was done while Shreyas Pai was a postdoctoral researcher at Aalto University. Supported in part by Research Council of Finland, Grant 334238, and Helsinki Institute for Information Technology HIIT.}

\end{center}

\newpage
\thispagestyle{empty}
\tableofcontents

\newpage
\pagenumbering{arabic}

\section{Introduction}
In this paper, we study the graph coloring problem in sparse graphs in a modern model of parallel computation, the \ampc-Model~\cite{Behnezhad2019-remote}. 
The Adaptive Massively Parallel Computation (\ampc) model of computation captures the challenges in modern platforms of processing massive data.
We have $P$ machines with $S = O(n^\delta)$ local space each, where $0 < \delta < 1$ and $n$ is the number of nodes in the input graph.
In synchronous rounds, the machines can communicate through a distributed key-value storage.
The machines have read access to the storage and, at the end of each round, the machines can write to the storage.
A crucial limitation is that the local space limits the number of reads and writes per round.
The \ampc model is a practically motivated extension of the standard \mpc model, where communication is restricted to all-to-all synchronous message-passing~\cite{KarloffSV10}. 

Before delving into our results, let us explain the necessary background.
The goal of the $k$-vertex-coloring problem is to assign a color out of $k$ colors to every node of a graph such that no neighboring nodes have the same color.
This is a classic and central problem that has been intensively studied for decades. 
Vertex-coloring also has close connections to other classic problems, such as the maximal independent set and scheduling.

In the distributed and parallel setting, the typical goal is to color the graph with $\Delta + 1$ colors, where $\Delta$ corresponds to the maximum degree of the input graph.
Indeed, any graph can be colored with $\Delta + 1$ colors, and for cliques and odd cycles, this is tight. 
However, many graphs can be colored with significantly fewer than $\Delta + 1$ colors. 
For example, sparse graphs can have a high maximum degree but can often be colored with much fewer colors.
Sparse graphs often occur in practice, which is causing a demand for better coloring algorithms for sparse graphs.

\paragraph{Arboricity.}
In this work, we study the vertex-coloring problem in graphs parameterized by \emph{arboricity}, which is a standard measure of ``everywhere'' sparsity.
The arboricity $\alpha$ corresponds to the minimum number of forests that the edges of a graph can be partitioned into~\cite{NashWilliams1964}. 
It is well known that any graph can be colored with $2\alpha$ colors, and in some cases this is tight.
As a straightforward example, a tree clearly has arboricity $1$ and can be colored with $2$ colors, while the maximum degree $\Delta$ can be up to $n - 1$.
Hence, arboricity captures the number of colors needed to color sparse graphs.

Finding \emph{arboricity-dependent} colorings with $O(\alpha)$, $\poly(\alpha)$, or $\exp(\alpha)$ colors has gained attention in various models of computation like \local \cite{BE10,barenboim10, barenboim15,fraigniaud16,MausT22,MausSPAA21}, \mpc \cite{Ghaffari2019-arboricity,Fischer2023,bera_et_alICALP2020}, dynamic algorithms \cite{HenzingerNW2020,Christiansen2023,ChristiansenRICALP2022, bhattacharya2023arboricity}, and streaming~\cite{bera_et_alICALP2020}.
A common approach is to find an acyclic orientation of the edges such that each node has an out-degree of $O(\alpha)$ or $\poly \alpha$.
Starting from sinks, the nodes can then be iteratively colored with out-degree + 1 colors.
Typically, the out-degree gives a trade-off between the runtime and the number of required colors.
This approach to coloring highlights a crucial difference between $(\Delta + 1)$-coloring and arboricity-dependent coloring and pinpoints a key challenge.
In the former case, any partial solution can be extended to a complete solution, but this is not the case for the latter.
In this light, the iterative coloring approach can be seen as carefully choosing partial solutions that can be extended to a complete solution.

In recent work, the problem has been studied in the context of dynamic algorithms~\cite{Christiansen2023}.
They showed that \underline{\smash{\emph{given} such an orientation,}} there is an efficient Local Computation Algorithm (\lca) to find an arboricity-dependent coloring. They also showed how to maintain such an orientation efficiently.
In the \lca model, each node can be queried for the $i$th entry in their adjacency list and their degree~\cite{lca11}. 
The goal is to implement an oracle that returns the output at any node $v$, say, the color of $v$ in the coloring problem. The crux is that the outputs (per node) must yield a consistent and feasible solution throughout the graph.

In our setting, the orientation is not given for free and one of our contributions is an \lca to find an acyclic orientation with a small out-degree.
As pointed out by the authors in~\cite{Christiansen2023}, there is a caveat to designing such an \lca.
Although there is no formal proof, it seems likely that the problem has a high volume, that is, any \lca needs to do linear (in $n$) queries in the worst case.
To circumvent this obstacle, we design an \lca that orients \emph{almost} all of the edges. We show that this \lca is sufficient for designing efficient \ampc algorithms to compute acyclic low out-degree orientations orienting \emph{all} edges. 
Lastly, we obtain various \ampc coloring algorithms that employ different trade-offs between their runtime and the number of colors used. We next detail these three contributions in order, (1) the \lca for computing ``weak orientations'', (2) \ampc algorithms for computing acyclic low out-degree orientations, and (3) \ampc coloring algorithms. 
We emphasize that all our algorithms are \emph{deterministic}.

\bigskip 

\noindent\textbf{Contribution 1: An \lca for computing a weak orientation.}
A $\beta$-partition of a graph is a partition of its vertex set into layers such that every node has at most $\beta$ neighbors in higher (or the same) layers. The size of the partition is the number of layers. This partition naturally gives rise to the desired orientation required for efficient graph coloring algorithms by orienting edges from smaller to larger layers while orienting them arbitrarily within the layers. In graphs with arboricity at most $\alpha$, it is known that such partitions with $O(\log n)$ layers exist when $\beta\geq (2+\eps)\alpha$ for any constant $\eps>0$ \cite{BE10}. Our core building block and also the largest challenge of our approach is the design of an \lca algorithm that computes a weak version of the desired partition. 
In \Cref{sec:nutshell} we underline the technical challenges associated with obtaining such a result. 
The following simplified version of the main lemma assumes that $\beta$ and $\alpha$ are constant and that $\beta\geq (2+\eps)\alpha$ holds for any $\eps>0$. We denote by $G[S]$ the subgraph induced by the node set $S$ in graph $G$.

\begin{lemma}[Simplified version of \Cref{lem:coinDropLCAformal}]
	\label{lem:coinDropLCA}
	For any constant $\delta>0$ there is deterministic \lca algorithm that uses at most $O(n^{\delta})$ queries per node on a graph $G$ with arboricity $\alpha$ and assigns each node a layer from $\mathbb{N}\cup \{\infty\}$ such that the following holds: 
	
	$\triangleright$ There exists a subset $S\subseteq V$  containing at least a $1-1/n^{O(\delta)}$ fraction of vertices such that the layering of $S$ forms a $\beta$-partition of $G[S]$ with $O(\log_{\beta} n)$ layers. $\triangleleft$
\end{lemma}

\bigskip
\noindent
\textbf{Contribution 2: \ampc algorithms for finding acyclic low out-degree orientations.} Our second contribution is designing low-space \ampc algorithms to find acyclic orientations in graphs with any arboricity $\alpha$.
Ultimately, our \ampc algorithm iteratively uses the above \lca to find a complete partition and the number of nodes that are not layered (after one application of the \lca) dictates the round complexity of the \ampc algorithm.

Essentially, \Cref{lem:coinDropLCA} can be applied recursively\footnote{This recursion does not work as stated, because we have omitted some details to simplify the exposition. We explain these subtleties in \Cref{sec:recursion-details}.} on the subgraph induced by the nodes that are assigned an $\infty$ layer, yielding $O(1/\delta)=O(1)$ recursion levels and implying the following theorem.

\begin{restatable}[$\beta$-partitioning]{theorem}{partitioning}
	\label{thm:partitioning}
	For any constant $\eps>0$ and any $\beta\geq (2+\eps)\alpha$ there is a deterministic \ampc algorithm that computes a $\beta$-partition of size $O(\log_{\beta/(2\alpha)}n)$ in $O(\log_{\beta/(2\alpha)}\beta)$ rounds on any $n$-node graph with arboricity $\alpha$.
	The algorithm requires $O(n^{\delta})$ local space and $O(n^{1 + \delta})$ total space for any constant $\delta > 0$. In particular,
	\begin{itemize}
		\item when $\beta=O(\alpha)$, we get size $O(\log n)$ in $O(\log \alpha)$ rounds and 
		\item when $\beta=O(\alpha^{1+\eps})$, we get size $O(\log_\alpha n)$ in $O(1)$ rounds.
	\end{itemize}
\end{restatable}

\bigskip
\noindent
\textbf{Contribution 3: Coloring sparse graphs in \ampc.}
We design \ampc algorithms for coloring with different trade-offs between their runtime and the number of colors used.

\begin{restatable}[Coloring Results]{theorem}{coloringthm}
	\label{thm:coloring}
	For any constant $\eps > 0$ on graphs with possibly non-constant arboricity $\alpha$, there are deterministic low-space \ampc algorithms to compute:
	\begin{enumerate}
		\item An $O(\alpha^{2+\eps})$-coloring in $O(1/\eps)$ rounds.
		\item An $O(\alpha^2)$-coloring in $O(\log \alpha)$ rounds.
		\item A $((2+\eps)\alpha + 1)$-coloring in $\tilde{O}(\alpha/\eps)$ rounds.\footnote{The $\widetilde{O}(\cdot)$ notation hides $O(\log \alpha)$ factors.}
	\end{enumerate}
\end{restatable}

All of our coloring algorithms heavily exploit that we first compute a small out-degree orientation using \Cref{thm:partitioning} with different parameters.
Although \Cref{thm:coloring} has various trade-offs between the runtime and the number of colors, its main strength lies in the following corollary. 
\begin{corollary}
	For any constant $\eps>0$  there is a constant time \ampc algorithm to compute a $((2+\eps)\alpha + 1)$-coloring on any graph with bounded arboricity $\alpha = O(1)$.
\end{corollary}

\paragraph{The Quadratic Barrier.}
One of our easiest results of \Cref{thm:coloring}, once we're given the small out-degree orientation, is the one coloring with quadratic number of colors. Here, we essentially simulate the \local model algorithm by Linial that, very efficiently, computes a $\beta^2$-coloring on a graph with an out-degree $\beta$ and only uses local information. The fact that the algorithm works in graphs with a small out-degree (instead of small maximum degree as usually cited) was first observed in \cite{BE10}.
Understanding when one can go below $o(\alpha^2)$ colors with low volume and/or round complexity seems fundamentally hard and is an intriguing open question (see, e.g., ~\cite{barenboimelkin_book, FHK16}).

\paragraph{Deterministic algorithms.}  We emphasize that, in contrast to most prior algorithms in the \ampc model and also the most efficient \lca algorithms for large degree graphs, all our algorithms are deterministic. 
We note that deterministic algorithms seem fundamentally more constrained than randomized ones.
As an example, in the deterministic setting, even on a tree, it seems out of reach to design a sublinear \lca that finds the nearest leaf of each node.
In a closely related \volume model\footnote{A crucial difference between \lca and \volume is that in \volume the queries must form a single connected component.}, it is known that $\Omega(n)$ deterministic queries are required to solve closely related problems on trees~\cite{Rosenbaum2020}.
Furthermore, there is an \emph{exponential} gap between the randomized and the deterministic complexities.

\paragraph{Very High Degrees.}
A particular challenge in massively parallel computation models is to efficiently deal with nodes whose degree is significantly larger than the local space of each machine, especially for deterministic algorithms. As we desire to have efficient algorithms for all ranges of parameters, we design an algorithm using orthogonal techniques. 
We note that the following theorem holds in the strictly weaker \cite{Behnezhad2019-remote} low-space massively parallel computing model (\mpc). 
\begin{restatable}[Deterministic Coloring]{theorem}{derandomizedcoloring}
	\label{thm:derandomizedColoring}
	For any $x > 1$, there is a deterministic low-space \mpc algorithm to compute an $2x\Delta$-coloring in any $n$-node graph with maximum degree $\Delta$ in $O(\log_x n)$ rounds using $O(n+m)\poly\log n$ total space.
\end{restatable}
In the case of $\alpha>n^{\delta}$, we can still compute the low out-degree partitions as claimed in \Cref{thm:partitioning}. Actually, the partition comes with the property that each layer also has a small maximum degree, and hence, we can apply \Cref{thm:derandomizedColoring} (with $x=\alpha^{O(\eps)} = n^{O(1)}$) to color each layer independently with $O(\alpha^{1+\eps})$ colors in $O(\log_x n) = O(1)$ time, using a separate color space for each layer. This is one of the ways, how we use \Cref{thm:derandomizedColoring} in the proof of \Cref{thm:coloring} in the case of huge arboricity.

\subsection{Related Work}
In the \lca model, there is an algorithm for the classic $(\Delta + 1)$-coloring problem with $\poly \Delta \cdot \log n$ queries~\cite{Chang2019}.
Given that the local space is larger than a small polynomial of $\Delta$, we can turn the \lca into a randomized \ampc algorithm that runs in constant time.
The current \emph{deterministic} state-of-the-art originates from the \local model of distributed message-passing~\cite{linial92}, where the edges of an input graph correspond to communication links along which the nodes communicate in synchronous rounds. 
Using the graph-exponentiation technique, we can simulate a known $\poly \log \log n$ -round algorithm~\cite{Chang2018} in $O(\log \log \log n)$ rounds.
This can be further derandomized to match the same runtime deterministically and similar techniques apply to $(\textrm{degree} + 1)$-coloring as well~\cite{Czumaj2021, Coy2024}.
In \ampc, we can do even better. 
The aforementioned algorithms consist of two parts: Solving the problem on low-degree instances in $O(\log \log \log n)$ time and high-degree instances in $O(\log^* n)$ time.
The former part can easily be simulated in $O(1)$ rounds in the low-space \ampc even though not explicitly stated in the papers.
The runtime then reduces to $O(\log^* n)$.

The coloring problem parameterized by the arboricity $\alpha$ has enjoyed attention in other related models of distributed and parallel computation as well.
In the \local model, Barenboim and Elkin gave deterministic $O(\log n)$ and $O(\alpha \log n)$ round algorithms for computing $O(\alpha^2)$ and $(2 + \eps)\alpha$-colorings, respectively. It is also known that $\Omega(\log n)$ rounds are required~\cite{linial92} for arboricity-dependent colorings. 
On the randomized side, it is known how to obtain an $O(\alpha \log \alpha)$ coloring in $O(\log n)$ rounds and $O(\alpha)$-coloring in $O(\log n \cdot \min\{ \log \alpha, \log \log n \})$ rounds~\cite{Lymouri2017}.
We note that there are also $\poly \log n$ time algorithms for finding $(1 + \eps)\alpha$-out-degree orientations, but these orientations need not be acyclic~\cite{HarrisSVPODC21}.

In the \clique model, we have $n$ machines that communicate in an all-to-all fashion. Each machine corresponds to a node in the input graph and initially knows its incident edges.
The \clique model is very close to the non-adaptive \mpc model with linear memory per machine, i.e., $\Omega(n)$ words per machine. There are deterministic algorithms for $O(\alpha)$-colorings in both of the models that run in $O(1)$ rounds~\cite{Fischer2023}.
Finally, through the rake-and-compress decomposition on trees, one can also obtain an orientation with an out-degree at most $2$.
In a recent work, it was shown how to obtain this decomposition in $O(1)$ rounds of \ampc~\cite{tree-contraction-2022}.
For the special case of forests $(\alpha=1)$, \cite{GrunauDISC2023treecoloring} give a $3$-coloring algorithm that runs in $O(\log \log n)$ rounds and uses optimal global space. This algorithm is conditionally optimal in the \mpc model.

\subsection{Outline}

In \Cref{sec:nutshell}, we start with giving a technical overview of our main contributions and the challenges we need to overcome. In \Cref{sec:preliminaries}, we define the \ampc model and also give several key definitions for structures like partial, induced, and natural $\beta$-partitions. In \Cref{sec:centralizedPartition}, we present the \lca for computing partial $\beta$-partition and define a coin dropping game which is useful in describing the query structure of the \lca, proving \Cref{lem:coinDropLCAformal} (formal version of \Cref{lem:coinDropLCA}). In \Cref{sec:AMPCPartition}, we show how to leverage the \lca to get a $\beta$-partition in the \ampc model, proving \Cref{thm:partitioning}. Finally, in \Cref{sec:coloring}, we show the coloring of sparse graphs, proving \Cref{thm:coloring,thm:derandomizedColoring}.

\section{Sublinear \lca: High Level Challenges and Techniques} \label{sec:nutshell}
We design a deterministic \lca with a sublinear number of queries to find an acyclic orientation with a small out-degree $\beta = \Omega(\alpha)$ for a very large fraction of the nodes in a graph with arboricity $\alpha$. The informal statement is \Cref{lem:coinDropLCA} and the formal one is \Cref{lem:coinDropLCAformal}.

The ultimate goal for our \lca would be to compute a so called \emph{$\beta$-partition}, which is a partition of its vertex set into layers such that every node has at most $\beta$ neighbors in higher (or the same) layers, with as few layers as possible. This partition naturally gives rise to the desired orientation required for efficient graph coloring algorithms by orienting edges from smaller to larger layers while orienting them arbitrarily within the layers. As discussed in the introduction, it seems unlikely that an \lca with sublinear query complexity can compute such a partition. Instead, we design a sublinear \lca that computes the layers of \emph{most} of the nodes in the graph, see \Cref{fig:partial}.

\begin{figure}[ht]
	\centering
	\includegraphics[width=8cm]{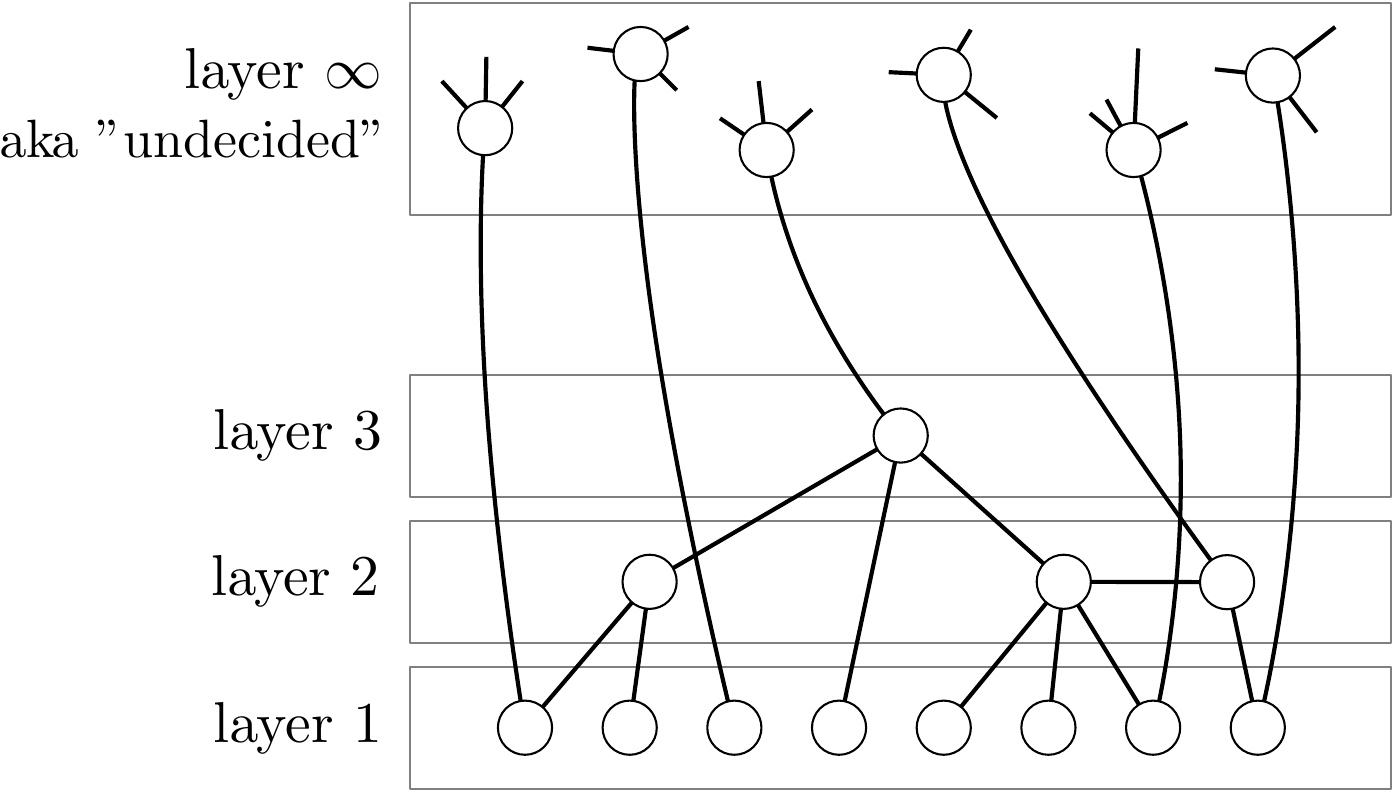}
	\captionof{figure}{An illustration of a $\beta$-partition where most of nodes are assigned a layer.}
	\label{fig:partial}
\end{figure}

\begin{figure}[ht]
	\centering
	\subfloat[The first three layers of a $\beta$-partition. Black nodes form the dependency graph $D(v)$ of node $v$.]{\includegraphics[width=.45\linewidth]{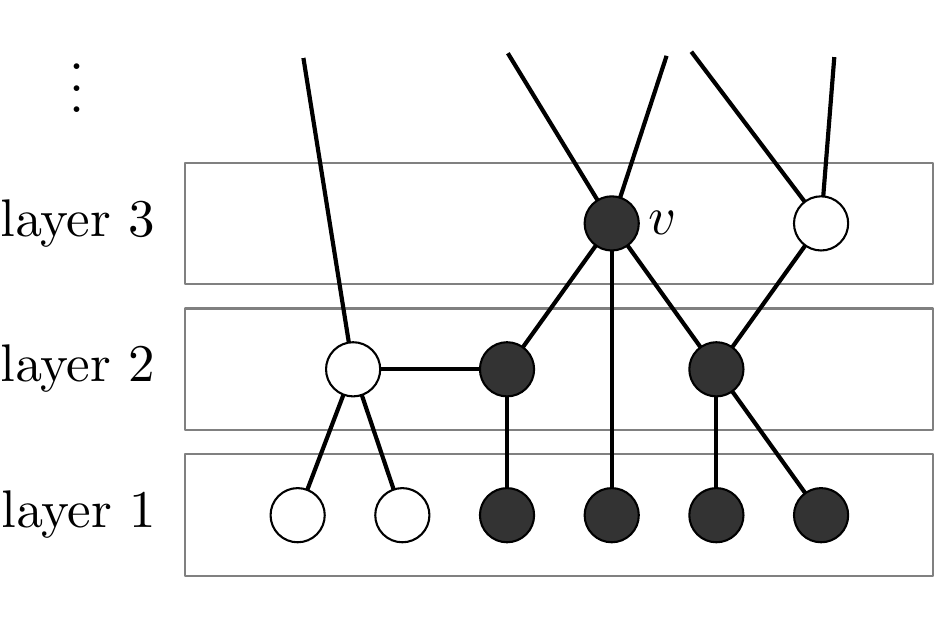} \label{fig:dg}} \hfill
	\subfloat[A skewed dependency graph, which is a counter-example to na\"ively volume-based querying.] {\includegraphics[width=.47\linewidth]{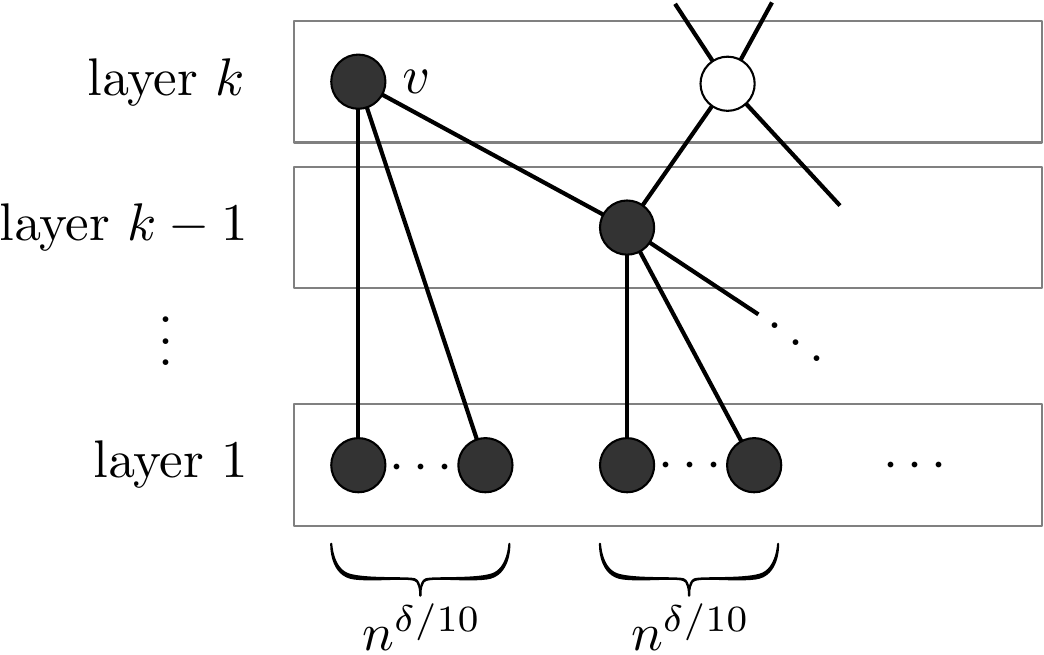} \label{fig:bad}} 
	\caption{Examples of different dependency graphs.}
\end{figure}

\subsection{Where should one query? }
A natural question to ask is what part of the graph is sufficient for a node to learn (or query), in order to be able to compute its layer? 
The question is easy to answer if we already knew the $\beta$-partition $\ell$ that we will compute. This $\beta$-partition implies  a subgraph $D(v):=G[D(v)]$ for each node $v$ that testifies the layer of node $v$ in $\ell$. In other words, the layer of node $v$ depends on $D(v)$, that's why we call it the \emph{dependency graph} of $v$. See \Cref{fig:dg} for an example. \textbf{Once node $v$ has learned $D(v)$ with its queries, it can compute its layer. }
More formally, the dependency graph $D(v)$ of a node $v$ contains all nodes $u$ that are reachable from $v$ via a path of nodes with strictly decreasing layers. Again, see \Cref{fig:dg}. Observe that every node in $D(v)$ has at most $\beta$ neighbors outside of $D(v)$. Of course, the imagined $\beta$-partition $\ell$ is something that we can utilize in the analysis and to build intuition, but we cannot use it to guide the nodes to explore the graph, simply because it is not known beforehand.

Observe that if the dependency graph of a node $v$ is too large, there is no hope for $v$ to learn it in a sublinear number of queries. Hence, our only hope is that at least nodes with small dependency graphs ($\leq n^\delta$) manage to learn their respective dependency graphs and compute their layers. In our analysis, we show that most of the nodes have small dependency graphs. 

Everything boils down to solving an impossible sounding task: \textbf{guide the graph exploration of a node $v$ such that most of what $v$ queries is within its dependency graph, with no prior information on where to search.} Deterministic graph exploration without prior knowledge is hard, especially since the graph may look identical in every direction at a node $v$, but only a small fraction of the directions contain the dependency graph and there is no way to identify these directions that we want to explore. Let us briefly explain why the classic search paradigms fail in this setting. 

\begin{itemize}
	\item \emph{Depth-first-search:} Recall that node $v$ can have up to $\beta$ neighbors outside the dependency graph. If the algorithm happens to initiate the DFS from a neighbor outside of $D(v)$, all of the queries might be used on learning something that is not the dependency graph.
	\item \emph{Breadth-first-search:} In order to avoid the pitfalls of DFS, one could employ a BFS with the hope of balancing the exploration of the graph. Indeed, after learning its $1$-hop surroundings, at least a $1/\beta$ fraction of what is queried is the dependency graph, which sounds promising. Unfortunately, this approach is also flawed. Consider a node $v$ and its neighbor $u$ that is outside $v$'s dependency graph. The degree of $u$ is unbounded, so if for example $\deg(u)=n^{\delta}$, most of the queries in the second level of the BFS will be spent on learning the neighbors of $u$, which are outside $D(v)$, and hence useless.
\end{itemize}

Clearly, a more sophisticated and fine-tune exploration procedure is required. Our actual approach explores the graph in a "volume-balanced" way, instead of "distance-balanced" like BFS and DFS.

\paragraph{Volume-based querying.}
In order to steer the direction of the queries, 
consider the following coin dropping procedure. Give $n^{\delta}$ coins to node $v$, who na\"ively (this approach will still not work) forwards the coins to all of its neighbors, giving each neighbor an equal share of the total number of $v$'s coins. Consider repeating this process recursively for every node holding coins, until coins cannot be further divided, which is in some sense a "volume-based" BFS. The forwarding of coins represents querying for nodes, and every node that has received a coin has been queried by $v$ and hence  \emph{known} to $v$. Optimally, all nodes of $D(v)$ receive a coin.
Unfortunately, this approach fails for a large family of counter-examples: skewed dependency graphs (\Cref{fig:bad}), comprising of a long path where every node has, e.g., $n^{\delta/10}$, neighbors in the lowest layer. In this case, $n^\delta$ coins will be already spent after only 10 steps of distributing coins, effectively learning only roughly $10 \cdot n^{\delta/10}$ dependency graph nodes.

\subsection{Our solution} 

Our solution uses the same coin dropping procedure, but we devise more sophisticated \emph{forwarding rules} that determine how to distribute the coins. These rules utilize the topology of the graph induced by $v$'s knowledge aka the nodes $v$ has discovered so far. The forwarding rules are also updated frequently during the course of the algorithm to steer the future exploration of the graph and making very strong usage of the adaptivity of the queries.

Recall that every node in $D(v)$ has at most $\beta$ neighbors outside of $D(v)$. Hence, when coins are at a node $u \in D(v)$ and $u$ forwards coins to $\beta+1$ neighbors, we ensure that at least one neighbor from $D(v)$ receives coins, or more precisely at least a $1/(\beta+1)$-fraction of the coins goes to neighbors in $D(v)$.
If we keep recursively applying this approach to all nodes that still have coins available, this roughly tells us, that after $k$ steps of forwarding, at least $x/(\beta+1)^k$ coins are at nodes that are contained in the dependency graph. Here $x$ is the number of coins that $v$ started with. However, it is not sufficient to choose $\beta+1$ neighbors arbitrarily, since it may be that we keep forwarding coins to nodes that we already know, i.e., nodes that have seen a coin before. In the end, it may be that all of these coins are just accumulated in a small part of $D(v)$ that we already know, or in other words, the process does not lead $v$ to querying new, undiscovered, parts of the dependency graph.

\paragraph{Our coin forwarding rules.}  Hence, we need to choose \emph{the correct neighbors} to whom we forward the coins such that enough coins "flow" into the part of the dependency graph we have not yet learned, and not too many "leak out" outside of the dependency graph or "get stuck" in a part of the dependency graph that we already know. The rules for choosing the correct neighbors are chosen adaptively and frequently updated. Node $v$ keeps simulating an estimation of the $\beta$-partition (this estimation is called the \emph{induced} $\beta$-partition, formalized in \Cref{def:inducedPartition}) for the whole subgraph it has queried so far, which we call $S_v$ (which is only a small part of the whole graph). In the simulation, every node in $S_v$ is assigned a layer by $v$, or if no layer is assigned, we assume it is $\infty$.  During an iteration, nodes forward coins to the $\beta+1$ neighbors with the largest simulated layer. If more than $\beta+1$ neighbors have an $\infty$ layer, a node can forward the coins to any such $\beta+1$ neighbors. We show that with this forwarding rule, enough coins traverse to parts of $D(v)$ that have not received any coins before, and we keep learning more and more nodes of $D(v)$. 

Surprisingly, we cannot halt the coin dropping game at the first point where node $v$ is able to compute a layer $\neq \infty$ for itself. It is possible, that after querying most (but not all) of its dependency graph, node $v$ computes its estimated layer to be, e.g., $n^{\delta/10}$, while its actual layer, if $v$ would have found the whole dependency graph, would be significantly smaller. This is unacceptable, since we ultimately want the number of layers to be logarithmic in $n$. The final success of our method depends on the frequency of updating the forwarding rules and the duration for which we run the algorithm.

\subsection{Using our \lca recursively}
\label{sec:recursion-details}

Due to the sublinear query complexity, our \lca can only compute the $\beta$-partition layers for a large fraction of nodes, but not all. One would hope to simply recurse on the subgraph remaining after removing all nodes with a $\neq \infty$ layer. However, this step turns out to require more care. Taking a closer look at  \Cref{lem:coinDropLCA}, we only get the guarantee that there exists a subset $S\subseteq V$ containing at least a $1-1/n^{O(\delta)}$ fraction of vertices such that the layering of $S$ forms a $\beta$-partition of $G[S]$. The difficulty of recursion stems from the fact that we cannot identify the nodes belonging to this magical set $S$ and that there may exist nodes outside of set $S$ that have a $\neq \infty$ layer, which are problematic, since including their layering in the partition may result in a \emph{globally inconsistent} partition. Consider nodes $v,u \not\in S$; for example, a node $v$ may compute layer $10$ for itself based on the fact that its neighbor $u$ should be in layer $9$. However, node $u$, for some reason, has computed its layer to be $11$. If we place nodes $v$ and $u$ in their preferred layers, node $v$ might end up with more than $\beta$ neighbors in higher layers. 

We resolve the issue with a slightly stronger \lca than stated in the informal \Cref{lem:coinDropLCA}. In addition to outputting a layer $\ell_u(u)$, each node $u\in V$ outputs a \emph{proof} of its own layer in the form of a $\beta$-partition $\ell_u$ on a small subgraph of $G$. The proof is basically a set of instructions for layering a part of the graph such that $v$ can obtain layer $\ell_u(u)$. 

Having access to this proof actually relieves us from needing to know which nodes having layer $\neq \infty$ belong to $S$ and which do not. We can merge all $\beta$-partitions $\ell_u$ that are given as proofs by nodes $u\in V$ with $\neq \infty$ via a global minimum function
\begin{align*}
	\ell(v)=\min_{u\in V}\ell_u(v)
\end{align*}

and get a globally consistent $\beta$-partition $\ell$ on the subgraph induced by all nodes with layer $\neq \infty$.
This is sufficient for using our querying routines for recursion in \ampc.

\section{Model and Definitons}
\label{sec:preliminaries}
\subsection{The \ampc model}

In the \ampc model, we have an input of size $N$, which is processed by $P$ machines, each with space $S$. The unit of measure is a word of size $O(\log n)$ bits. The total space of all machines is $T = S \cdot P$ and we assume that $S = \Theta(N^{1-\Omega(1)})$.
In addition, there is a collection of distributed data stores (DDS) that we denote by $D_0,D_1,\dots$. We assume that all DDS provide a key-value store semantics: they store a collection of key-value pairs, such that given a key, a DDS returns the corresponding value. We require that both key and value consist of a constant number of words. The input data is stored in $D_0$ and uses a set of keys known to all machines (e.g. consecutive integers). The solution to a problem is stored in the last data store.

If there are $k > 1$ key-value pairs stored in the DDS with the same key $x$, the individual values can be accessed by querying for keys $(x,1),\dots,(x,k)$. Note that the indices from $1$ to $k$ are assigned arbitrarily. A query for a key that does not occur in the DDS, results in an empty response. The computation proceeds in rounds. In the $i$th round, each machine can read data from $D_{i-1}$ and write to $D_i$. Within a round, each machine can perform $O(S)$ reads (henceforth called queries) and $O(S)$ writes, and carry out arbitrary computation. Each query refers to querying for a single key-value pair from $D_{i-1}$ and each write refers to writing a single key-value pair to $D_i$. The amount of communication that a machine performs per round is equal to the total number of queries and writes. The crucial property of the model is that the queries that a machine makes within a round may depend on each other. In particular, say $g$ is a function from $X$ to $X$ and for each $x \in X$, $D_{i-1}$ stores a key-value pair $(x,g(x))$. Then, in round $i$ a machine can compute $g^k(y)$ if $k = O(S)$.

The realism of the \ampc model is discussed in more depth in \cite{Behnezhad2019-remote}. The algorithms presented in this paper take a graph with $n$ nodes and $m$ edges as an input. In this case, the input size is $N=n+m$. We also assume that $S = n^\delta$, where $\delta \in (0,1)$ is some constant. 

\subsection{Arboricity}

\begin{definition}[Arboricity]\label{def:arboricity}
	The arboricity $\alpha(G)$ of a graph $G= (V, E)$ is defined as
	\[
	\alpha(G) = \max_{\substack{G' \subseteq G, \\ |V(G')| \geq 2}} \left\{\left\lceil \frac{|E(G')|}{|V(G')|-1} \right\rceil \right\}.
	\]
\end{definition}

An equivalent definition for $\alpha(G)$ is the minimum number of edge-disjoint forests into which the edge set of $G$ can be partitioned~\cite{NashWilliams1964}. In this work, we will consider graphs with arboricity bounded by parameter $\alpha$ that is $\alpha(G) \le \alpha$ for some parameter $\alpha$. To simplify notation we will assume that the arboricity of the graph is $\alpha$. We now state some simple facts about the arboricity that immediately follow from \Cref{def:arboricity} and were also observed in \cite{BE10}.

\begin{fact}[\cite{BE10}] \label{fact:ArbDec}
	For any subgraph $G'$ of $G$, the arboricity of $G'$ is at most the arboricity of $G$.
\end{fact}

\begin{fact}[\cite{BE10}] \label{fact:manySmallDegreeNodes}
	For any $\eps > 0$, a graph $G=(V,E)$ with arboricity $\alpha$ has at least $\eps|V|/(2+\eps)$ vertices with degree at most $(2+\eps)\cdot \alpha$. 
\end{fact}

An equivalent formulation of \Cref{fact:manySmallDegreeNodes} is that, for any $\eps > 0$, a graph with arboricity $\alpha$ has less than $2|V|/(2+\eps)$ vertices with degree greater than $(2+\eps)\cdot \alpha$. We generalize this statement by parameterizing the degree bound as follows.

\begin{lemma}[Generalization of \Cref{fact:manySmallDegreeNodes}]  \label{lem:generalManySmallDegreeNodes}
	For any $\beta \geq 1$, a graph $G=(V,E)$ with arboricity $\alpha$ has less than $2\alpha|V|/\beta$ vertices with degree greater than $\beta$. 
\end{lemma}

\begin{proof}
	Suppose for a contradiction that there are at least $2\alpha|V|/\beta$ vertices with degree greater than $\beta$. It follows that,
	\begin{align*}
		2|E| = \sum_{v \in V} \deg(v) > \beta \cdot \frac{2\alpha \cdot |V|}{\beta} = 2 \alpha \cdot |V| \stackrel{*}{\geq} 2 \cdot \frac{|E|}{|V|-1} \cdot |V| > 2 |E|~,
	\end{align*}
	which is a contradiction. The inequality $*$ follows from \Cref{def:arboricity}, as $|E|/(|V|-1)$ must be a lower bound for arboricity.
\end{proof}

\subsection{$\beta$-Partition}
We begin with generalizing the graph-theoretic structure known as \emph{$H$-partition}, which were introduced in the context of distributed message passing models by Barenboim and Elkin \cite{BE10}. Informally, they define an $H$-partition as a partition of a graph with arboricity $\alpha$ into $O(\log n)$ layers such that every node has at most $(2+\eps)\cdot\alpha$ neighbors in higher or equal layers, where $n$ is the number of nodes in the graph.

For our purposes, we define a \emph{$\beta$-partition}, which is a partition of a graph into layers such that every node has at most $\beta$ neighbors in higher or equal layers, for any positive value $\beta$. Our definition also permits an infinity layer: nodes in the infinity layer can have an unbounded number of neighbors in the infinity layer. If the infinity layer is non-empty, we call the partition a \emph{partial $\beta$-partition}. Otherwise, we simply call it a $\beta$-partition.

\begin{definition}[(partial) $\beta$-partition] \label{def:hPartition}
	Consider a graph $G=(V,E)$ and a function $\lambda \colon V(G) \arr \mathbb{N} \cup \{\infty\}$. For $i \in \mathbb{N} \cup \{\infty\}$ define  $V_i=\{v\in V(G) \mid \lambda(v)=i\}$ and $V_{\geq i} = \bigcup_{j\geq i} V_j$.
	
	We say that $\lambda$ is a \emph{$\beta$-partition} if $\deg_{V_{\geq i}}(v) \leq \beta$ for every node $v \in G$ with $\lambda(v) \neq \infty$ and $|V_{\infty}|=0$. If $|V_{\infty}| > 0$, we call $\lambda$ a \emph{partial $\beta$-partition}. 
	We refer to $V_i$ as \emph{layer} $i$ of the $\beta$-partition and $\lambda(v)$ as the \emph{layer} of a node $v$. The \emph{size} of a (partial) $\beta$-partition is the number of unique non-infinity layers.	
\end{definition}

Observe that if $\beta=n$ on a $n$-node graph, the $\beta$-partition would just contain one layer -- the input graph. Hence, $\beta$-partitions are useful only for certain ranges of $\beta$. For example, when $\beta=(2+\eps)\alpha$ where $\alpha$ is the arboricity, we get the well known $H$-partition by Barenboim and Elkin \cite{BE10}.

We extend the definition of $\beta$-partitions further to node-induced subgraphs, where it is a priori unclear how to deal with neighbors not in the subgraph.

\subsection{Induced $\beta$-partition}

Barenboim and Elkin \cite{BE10} define an algorithm where we iteratively put all nodes with degree at most $(2+\eps)\alpha$ in the top-most layer and delete them from the graph. This algorithm has the property that each node will have at most $(2+\eps)\alpha$ neighbors in the same layer or above, and therefore, the output of this algorithm is a $(2+\eps)\alpha$-partition. Using \Cref{fact:ArbDec,fact:manySmallDegreeNodes}, it is easy to see that size of the partition produced by this algorithm is at most $O(\log_{(2+\eps)/2} n)$.
In \Cref{def:inducedPartition}, we generalize the algorithm of \cite{BE10} in two ways: first, by making the degree bound $\beta$ and second, by running it only on a subset $S$ of the nodes.

In the following, we emphasize that all notions refer to the original graph $G$ if not specified otherwise, even when we consider an (induced) subgraph of $G$. In particular, $N(u)$ and $\deg(u)$ always refer to $N_G(u)$ and $\deg_G(u)$, respectively. The following definitions do not assume any particular model of computation.

\begin{definition}[$S$-induced $\beta$-partition] \label{def:inducedPartition}
	Consider a graph $G=(V,E)$, a subset $S \subseteq V$ and any $\beta\geq1$. Consider the partial $\beta$-partition $\sigma_{S,\beta} \colon V(G) \arr \mathbb{N} \cup \{\infty\}$ constructed as follows.
	\begin{enumerate}
		\item Set $\sigma_{S,\beta}(v) \larr \infty$ for every $v \in V$.
		\item For $i=0,\dots,|S|$:
		
		Set $\sigma_{S,\beta}(v) \larr i$ for all nodes $v \in S$ with at most $\beta$ nodes $u \in N(v)$ for which $\sigma_{S,\beta}(u) = \infty$.
	\end{enumerate}
	We refer to $\sigma_{S,\beta}(v)$ as the \emph{layer} of a node $v$.
\end{definition}

We often write $\sigma$ or $\sigma_{S}$ instead of $\sigma_{S,\beta}$ when subset $S$ and/or parameter $\beta$ are clear from context. The following lemma proves many useful properties about function $\sigma_{S,\beta}$ produced by \Cref{def:inducedPartition}. In particular, part \textit{ii.} proves that $\sigma_{S,\beta}$ is indeed a partial $\beta$-partition.

\begin{lemma} \label{lem:inducedPartitionProp}
	Consider a graph $G=(V,E)$, a subset $S \subseteq V$ and any $\beta\geq1$. The partial $\beta$-partition $\sigma_{S,\beta}$ from \Cref{def:inducedPartition} satisfies the following properties for a node $v \in S$.
	\begin{itemize}
		\item[i.] If $\sigma_{S,\beta}(v) = \infty$, then $\sigma_{S,\beta}(w) = \infty$ for at least $\beta+1$ nodes $w \in N(v)$
		\item[ii.] If $\sigma_{S,\beta}(v) \neq \infty$, then $\sigma_{S,\beta}(w)\geq\sigma_{S,\beta}(v)$ for at most $\beta$ nodes $w \in N(v)$
		\item[iii.] If $\sigma_{S,\beta}(v)\neq \infty$ and $\deg(v) \geq \beta+1$, then  $\sigma_{S,\beta}(w) \geq \sigma_{S,\beta}(v)-1$ for at least $\beta+1$ nodes $w \in N(v)$
	\end{itemize}
\end{lemma}

\begin{proof} In the following, an iteration refers to one round of Step 2 in \Cref{def:inducedPartition}. Let $\sigma_S^k(v)$ denote value $\sigma_S(v)$ at the beginning of iteration $k$. For readability, we write $\sigma$ instead of $\sigma_{S,\beta}$.
	
	\textit{i.} Observe that if there is an iteration where no node is assigned a $\sigma$ value, none of the following iterations assign $\sigma$ values. Hence, if value $|S|$ is assigned, $\sigma(v) \neq \infty$ holds for all nodes $v \in S$. Towards a contradiction, consider an iteration $i<|S|$ after which $\sigma(v) = \infty$ and there are at most $\beta$ nodes $w \in N(v)$ with $\sigma(w)=\infty$. In iteration $i+1$, it must be that $\sigma(v) \larr i+1 \neq \infty$, a contradiction.
	
	\textit{ii.} Consider $\sigma(v) = \sigma^i(v) = i \neq \infty$. By definition, it holds that $\sigma^i(w)=\infty$ for at most $\beta$ neighbors $w \in N(v)$. Hence, for said neighbors it will hold that $\sigma(w)=\infty$ or $\sigma(w) \in [i,|S|]$. In either case, $\sigma(w) \geq \sigma(v)$.
	
	\textit{iii.} Consider $\sigma(v) = \sigma^i(v) = i \neq \infty$. Since $\deg(u) \geq \beta+1$, it must hold that $i>0$. Define $I_{i}$ to be the set of nodes $w \in N(v)$ for which $\sigma^{i}(w)=\infty$. Define $I_{i-1}$ similarly and observe that $I_{i} \subsetneq I_{i-1}$: since $\sigma(v) = i$, it must be that $|I_{i-1}|\geq \beta+1$ and $|I_{i}| \leq \beta$. Hence, for nodes $w \in I_{i-1} \setminus I_{i}$ it holds that $\sigma(w)=i-1$. Since $\sigma^i(w)=\infty$ for nodes $w \in I_i$, it must be that $\sigma(w)\geq i$. Hence, for $|I_{k-1}|\geq \beta+1$ nodes $w \in N(v)$ it holds that $\sigma(w) \geq i-1 = \sigma(v)-1$.
\end{proof}

The following lemma reveals the so called non-increasing nature of the $S$-induced $\beta$-partition. Very informally, the intuition for the lemma is that the larger the subgraph used to construct the $S$-induced $\beta$-partition, the smaller the layer assignment for every node in the graph. In other words, more information gives a better partition.

\begin{lemma} \label{lem:sigmaDec}
	Consider a graph $G=(V,E)$, subsets $S \subseteq T \subseteq V$ and any $\beta\geq1$. Then $\sigma_{S,\beta}(v)\geq \sigma_{T,\beta}(v)$ for every node $v \in V$.
\end{lemma}

\begin{proof}
	For readability, we write $\sigma_{S}$ and $\sigma_{T}$ instead of $\sigma_{S,\beta}$ and  $\sigma_{T,\beta}$, respectively. If $S=T$, then $\sigma_{S}$ equals $\sigma_{T}$ by definition. Hence, consider only $S \subset T$. The claim holds for all nodes $v$ for which $\sigma_S(v)=\infty$. Since $\sigma_S(v) = \infty$ holds for all nodes $v \in V \setminus S$ by definition, we only need to prove the claim for nodes $v \in S$ for which $\sigma_S(v)\neq\infty$. 
	
	Consider the construction of the induced partitions $\sigma_S$ and $\sigma_T$ from \Cref{def:inducedPartition} side-by-side. Towards a contradiction, let $i$ be the first iteration (side-by-side) after which $\sigma_S(v)<\sigma_T(v)$ holds for a node $v \in S$. It must be that $\sigma_S(v) = i \neq \infty$ and $\sigma_T(v) = \infty$, which implies that $\sigma_S(w)=\infty$ for at most $\beta$ nodes $w \in N(v)$ and that $\sigma_T(v)=\infty$ for at least $\beta+1$ nodes $w \in N(v)$. Hence there must have existed one node $x$ for which $\sigma_S(x)<\sigma_T(x)$ after iteration $i-1$, a contradiction.
\end{proof}

The analysis of our algorithm will crucially rely on the properties of a so called \emph{dependency graph} of a node, which is defined through a (partial) $\beta$-partition.

\begin{definition}[Dependency graph] \label{def:dependencyGraph}
	Consider a graph $G=(V,E)$, a subset $S \subseteq V$ and any $\beta\geq1$. The dependency graph $G[D(\sigma_{S,\beta},v)]$ of a node $v\in V$ is defined as the output of the following recursive procedure. 
	\begin{itemize}
		\item If $\sigma_{S,\beta}(v)=\infty$, then $D(\sigma_{S,\beta},v) = \emptyset$ 
		\item If $\sigma_{S,\beta}(v)=0$, then $D(\sigma_{S,\beta},v) = \{v\}$
		\item otherwise $D(\sigma_{S,\beta},v) =  \left\{ \bigcup_{w \in N^{<}(v)} D(\sigma_{S,\beta},w) \right\} \cup \{v\}$ where $N^{<}(v)$ is the set of nodes $w \in N(v)$ such that $\sigma_{S,\beta}(w)<\sigma_{S,\beta}(v)$.
	\end{itemize}
\end{definition}

\begin{observation}[Nested property of dependency graphs] \label{obs:nestedDepG}
	Consider a graph $G=(V,E)$, a subset $S \subseteq V$ and any $\beta\geq1$. For all $w \in D(\sigma_{S,\beta},v)$ for a node $v \in S$ it holds that $D(\sigma_{S,\beta},w) \subseteq D(\sigma_{S,\beta},v)$.
\end{observation}

\begin{proof}
	By \Cref{def:dependencyGraph}, a node $w$ can be added to $D(\sigma_{S,\beta},v)$ only by adding $D(\sigma_{S,\beta},w)$ to $D(\sigma_{S,\beta},v)$, and $D(\sigma_{S,\beta},w)$ contains $w$ by definition.
\end{proof}

\begin{lemma} \label{lem:lessThanBeta}
	Consider a graph $G=(V,E)$, a subset $S \subseteq V$ and any $\beta\geq1$. If $\sigma_{S,\beta}(v) \neq \infty$ for a node $v \in S$, then $|N(v) \setminus D(\sigma_{S,\beta},v)| \leq \beta$. 
\end{lemma}

\begin{proof}
	We write $\sigma$ instead of $\sigma_{S,\beta}$. By \Cref{def:dependencyGraph}, set $N(v) \cap D(\sigma,v)$ contains all nodes $w \in N(v)$ for which $\sigma(w)< \sigma(v)$. Hence, $N(v) \setminus D(\sigma,v)$ contains all nodes for which $\sigma(w) \geq \sigma(v)$. Because $\sigma(v)\neq \infty$, by \Cref{lem:inducedPartitionProp} it holds that $\sigma(w)\geq \sigma(v)$ for at most $\beta$ nodes $w \in N(v)$.
\end{proof}

\subsection{Natural $\beta$-partition}

We want to accentuate the case when the whole node set of a graph is used to construct an induced $\beta$-partition. For a graph $G=(V,E)$, we call the $V$-induced $\beta$-partition the \emph{natural $\beta$-partition}. 

\begin{definition}[Natural $\beta$-partition] \label{def:naturalHpartition} 
	Consider a graph $G=(V,E)$, and any $\beta\geq1$. The natural $\beta$-partition $\ell_\beta$ is defined as the $V$-induced $\beta$-partition $\sigma_{V,\beta}$.
\end{definition}

For $\beta=(2+\eps)\alpha$ the natural $\beta$-partition corresponds to the $H$-partition from \cite{BE10}. 
Our motivation behind the natural $\beta$-partition definition is that $\ell_\beta$ assigns the lowest possible layer to all nodes in the graph among all induced $\beta$-partitions. More formally, we have the following lemma.

\begin{lemma} \label{lem:naturalBest}
	Consider a graph $G=(V,E)$, a subset $S \subseteq V$ and any $\beta\geq1$. It holds that $\sigma_{S,\beta}(v) \geq \ell_{\beta}(v)$ for every $v \in V$. 
\end{lemma}

\begin{proof}
	By \Cref{lem:sigmaDec}, it holds that $\sigma_{S,\beta}(v) \geq \sigma_{T,\beta}(v)$ for any subgraphs $S \subseteq T \subseteq V$ and every node $v \in V$. When $T=V$, it holds that $\sigma_{S,\beta}(v) \geq \sigma_{V,\beta}(v)=\ell_\beta(v)$.
\end{proof}

Later, in our algorithmic usage of induced $\beta$-partitions, we will compute many of these partitions in parallel on various induced subgraphs $G[S]$, $S\subseteq V$ of the input graph $G$. The next lemma states that the layer assignment of a node $v$ in $\sigma_{S,\beta}$ corresponds to its layer assignment in the natural $\beta$-partition if $G[S]$ contains the dependency graph of $v$.
\begin{lemma} \label{lem:DGinS}
	Consider a graph $G=(V,E)$, a subset $S \subseteq V$ and any $\beta\geq1$. If $D(\ell_\beta,v) \subseteq S$ for a node $v \in V$, then $\sigma_{S,\beta}(w)=\ell_\beta(w)$ for every $w \in D(\ell_\beta,v)$.
\end{lemma}

\begin{proof}
	For readability, we write $\sigma_{S}$ and $\ell$ instead of $\sigma_{S,\beta}$ and $\ell_{\beta}$, respectively.
	
	Let us first prove the claim for $S=D(\ell,v)$ by induction on the values of $\ell(v)$, the base case being $\ell(w)=0$. The base case implies that $\deg(w)\leq \beta$, so by \Cref{def:inducedPartition} it holds that $\sigma_S(w)=0$. Assume that $\sigma_S(w)=\ell(w)$ for all nodes $w \in S$ with $\ell(w) < i$. Consider a node $u \in S$ with layer $\ell(u)=i$. Observe that all nodes $w \in N(u)$ with $\ell(w)<\ell(u)=i$ are in $S$ and $\sigma_S(w)=\ell(w)$ by induction. Since $\ell(u)=i$, by \Cref{lem:inducedPartitionProp}, there are at most $\beta$ nodes $w \in N(u)$ with $\ell(w) \geq \ell(u)$, implying that $\sigma_S(u)$ will be $i$.
	
	We have established that $\sigma_{D(\ell,v)}(w)=\ell(w)$ for every node $w \in D(\ell,v)$ if $S=D(\ell,v)$. If $S \supset D(\ell,v)$, the previous implies $\sigma_S(w) \leq \ell(w)$, because by \Cref{lem:sigmaDec} it holds that $\sigma_S(w) \leq \sigma_{D(\ell,v)}(w)$. On the other hand, by \Cref{lem:naturalBest}, for any subgraph $S \subseteq G$ it holds that $\sigma_S(w) \geq \ell(w)$ for every node $w \in S$. Hence, if $D(\ell,v) \subseteq S$ it holds that $\sigma_S(w)=\ell(w)$ for every node $w \in D(\ell,v)$.
\end{proof}

\section{An \lca for Partial $\beta$-partition} \label{sec:centralizedPartition}

In this section, we present an \lca for constructing a partial $\beta$-partition (\Cref{def:hPartition}) of the input graph using a coin dropping game. The game is formulated from the perspective of a single node making local queries to the graph and can be easily implemented as an \lca (see \Cref{lem:coinDropLCAformal}). Later, we will simulate this game for many nodes in parallel in the \ampc model (see \Cref{sec:AMPCPartition}).

\subsection{The Coin Dropping Game} \label{sec:coinDroppingGame}

The high-level idea of the coin dropping game for a node $v$ in graph $G=(V,E)$ is as follows. During the game, we keep track of a certain subset of nodes $S_v \subseteq V$ that initially contains only node $v$, but grows in every super-iteration. In the beginning of every super-iteration, we compute the $S_v$-induced $\beta$-partition $\sigma_{S_v,\beta}$ from \Cref{def:inducedPartition}. One can think of $\sigma_{S_v,\beta}(v)$ as an approximation of the natural $\beta$-partition $\ell_{\beta}=\sigma_{V,\beta}$ (see \Cref{def:naturalHpartition}). During a super-iteration, we perform a coin distribution process (consisting of iterations) that uses $\sigma_{S_v,\beta}$ as a guiding mechanism to add certain nodes to $S_v$ such that $\sigma_{S_v,\beta}(v)$ of the next super-iteration is a ``better'' approximation of  $\ell_{\beta}(v)$.

The coin distribution process consist of iterations, in each of which, nodes holding coins forward them along the edges of the graph according to certain \emph{local forwarding rules}. In the very first iteration, only node $v$ is given coins. The whole process can be viewed as a flow of coins that starts at node $v$ and spreads through the graph. The flow stops when it reaches nodes outside of $S_v$, after which all nodes outside of $S_v$ that have received a coin are added to $S_v$ of the next super-iteration. The local forwarding rules manifest as \emph{forwarding sets} for every node in $S_v$, which require computing the $S_v$-induced $\beta$-partition. Informally, the forwarding set of a node $u$ is a small subset of $u$'s neighbors, to whom $u$ will be forwarding its coins during the super-iteration. 
More formally, see \Cref{def:forwardingSet}.
While keeping the size of set $S_v$ roughly below $x^3$ for some parameter $x$, we ensure that eventually $\sigma_{S_v,\beta}(v)=\ell_{\beta}(v)$ for all nodes that have small enough dependency graphs, i.e., $|D(\ell_{\beta},v)|<x^2$ and a small enough layer in the natural $\beta$-partition, i.e., $\ell_{\beta}(v)\leq \log_{\beta+1} x$. We emphasize that the game does not guarantee that $D(\ell_{\beta},v) \subseteq S_v$, which complicates the analysis.

\begin{definition}[Forwarding set] \label{def:forwardingSet}
	Consider a graph $G=(V,E)$, a subset $S \subseteq V$ and any $\beta\geq1$. Let $\sigma_{S,\beta}$ be the $S$-induced $\beta$-partition. The forwarding set $F(\sigma_{S,\beta},u)$ for a node $u \in S$ is the set of any $\min\{\deg(u),\beta+1\}$ nodes $w \in N(u)$ with the highest $\sigma_S(w)$ values. 
\end{definition}

\paragraph{$(x,\beta,F)$-coin dropping game.} Consider executing the $(x,\beta,F)$-coin dropping game for a node $v$ in a graph $G=(V,E)$. Throughout the game, we keep track of a subset $S_v \subseteq V$. The game consists of $x^2$ super-iterations, where one super-iteration proceeds as described in \Cref{alg:coinDropping}. In the first super-iteration $S_v$ is initialized to $S_v \larr \{v\}$.

\begin{algorithm}[H]
	\caption{A super-iteration of the $(x,\beta,F)$-coin dropping game for a node $v \in V$.}\label{alg:coinDropping}
	\begin{algorithmic}[1]
		\vspace{1mm}
		\State{Compute the $S_v$-induced $\beta$-partition $\sigma_{S_v,\beta}$ and the forwarding set $F(\sigma_{S_v,\beta},u)$ for all $u \in S_v$.\vspace{1mm}}
		\State{Node $v$ is given $x$ coins.\vspace{1mm}}
		\State{For $|V|$ iterations: every node $u \in S_v$ with $x'\geq |F(\sigma_{S_v,\beta},u)|$ coins forwards $x'/|F(\sigma_{S_v,\beta},u)|$ coins to every node in $F(\sigma_{S_v,\beta},u)$. \vspace{1mm}} 
		\State{Every node $u \not\in S_v$ with coins is added to $S_v$.} 
	\end{algorithmic}
\end{algorithm}

\subsection{Game Analysis} \label{sec:gameAnalysis}

Let us analyze the coin dropping game. Recall that \Cref{alg:coinDropping} refers to one super-iteration of the game. Recall that nodes cannot compute the natural $\beta$-partition $\ell_{\beta}$ themselves. However, we utilize $\ell_{\beta}$ for analysis. The following lemma ensures progress, and informally it states that if the $S_v$-induced $\beta$-partition assigns a layer to node $v$ which does not equal $\ell_{\beta}(v)$ (it can be anything between $\ell_{\beta}(v)$ and $\infty$), a node from the dependency graph $D(\ell_{\beta},v)$ is added to $S_v$ which was not in $S_v$ previously.

\begin{lemma}[Measure of progress] \label{lem:measureOfprogress}
	Consider a graph $G=(V,E)$, any $\beta,x\geq1$ and the natural $\beta$-partition $\ell_{\beta}$. Consider the beginning of an arbitrary super-iteration of the $(x,\beta,F)$-coin dropping game for a fixed node $v \in V$ and the corresponding set $S_v \subseteq V$. If $|D(\ell_{\beta},v)|\leq x^2$, $\ell_{\beta}(v) \leq \log_{\beta+1} x$ and $\sigma_{S_v,\beta}(v) > \ell_{\beta}(v)$, a node $w \in D(\ell_{\beta},v) \setminus S_v$ is added to $S_v$ at the end of the super-iteration.
\end{lemma}
\begin{proof}
	For readability, we write $\sigma$ and $\ell$ instead of $\sigma_{S_v,\beta}$ and $\ell_{\beta}$, respectively. Consider a node $v \in V$ as in the lemma statement. Recall that $v \in S_v$ by definition. Since $\sigma(v)=\infty$ by the lemma statement, we can assume that $\deg(v)\geq \beta +1$ (otherwise $\sigma(v)=0=\sigma_{V}(v)$ by definition).  Let $P=(w_0,w_1,\dots,w_k)$ be the longest path path where 
	\begin{align*}
		w_0=v \text{ and } w_{i} \in F(\sigma,w_{i-1}) \cap D(\ell,w_{i-1})
	\end{align*}
	
	By the guarantees of \Cref{lem:recursion}, for every $w_i$ it holds that $\sigma(w_i) > \sigma(w_i)$. Observe that by \Cref{obs:nestedDepG}, it holds that $P \subseteq D(\ell,v)$. Because $P$ cannot be any longer, \Cref{lem:recursion} cannot be applied to $w_k$, so either $w_k \not\in S_v$ or $\deg(w_k) \leq \beta$. There are two possible cases.

	\begin{itemize}
		\item $w_k \in S_v$ and $\deg(w_k) \leq \beta$: This implies $\sigma(w_k)=0$, contradicting $\sigma(w_k) > \ell(w_k)$, since $\ell$ cannot be negative.
		\item $w_k \not\in S_v$: Because $w_i \in D(\ell,w_{i-1})$, it holds that $\ell(w_{i-1})>\ell(w_i)$ by definition for every $i \in [1,k]$. Since $\ell(v) \leq \log x$ by the lemma statement, it must hold be $k \leq \log x$. In \Cref{alg:coinDropping}, $v (=w_0)$ forwards $x/(\beta +1)$ coins to $w_1$, who forwards $x/(\beta +1)^2$ coins to $w_2$, and so on. Since $k \leq \log x$, node $w_k$ will receive at least $x/(\beta+1)^{\log_{\beta+1} x} = 1$ coin(s) and will be added to $S_v$ at the end of the super-iteration. Because $w_k \not\in S$ and $w_k \in P \subseteq D(\ell,v)$, it holds that $w_k \in D(\ell,v) \setminus S_v$, completing the proof. \qedhere
	\end{itemize}
\end{proof}

The following lemma was used in \Cref{lem:measureOfprogress}. Informally, it ensures that roughly a $1/\beta$ fraction of coins forwarded by a node $u \in S_v$ are forwarded to nodes in $D(\ell_{\beta},u)$.

\begin{lemma} \label{lem:recursion}
	Consider a graph $G=(V,E)$, any $\beta,x\geq1$ and the natural $\beta$-partition $\ell_{\beta}$. Consider the $(x,\beta,F)$-coin dropping game for a fixed node $v \in G$ and the corresponding set $S_v \subseteq V$ during any super-iteration. Let $u \in S_v$ be a node with $\deg(u) \geq \beta+1$. If $\sigma_{S_v,\beta}(u) > \ell_{\beta}(u)$, then there exists at least one node $w \in F(\sigma_{S_v,\beta},u) \cap D(\ell_{\beta},u)$ for which $\sigma_{S_v,\beta}(w) > \ell_{\beta}(w)$.
\end{lemma}

\begin{proof}
	For readability, we write $\sigma$ and $\ell$ instead of $\sigma_{S_v,\beta}$ and $\ell_{\beta}$, respectively. Consider $u \in S_v$ as in the lemma statement. Since $\deg(u) \geq \beta+1$, it holds that $|F(\sigma,u)|=\beta+1$ by \Cref{def:forwardingSet}. Since $\ell(u) \neq \infty$, it holds that $|N(u) \setminus D(\ell,u)| \leq \beta$ by \Cref{lem:lessThanBeta}. Recall that $\ell(w) < \ell(u)$ holds for every node $w \in D(\ell,u) \setminus \{u\}$ by \Cref{def:dependencyGraph}. Hence, it holds that $\ell(w) < \infty$ for every node $w \in D(\ell,u) \setminus \{u\}$. 
	
	Let us consider the case where $\sigma(u)=\infty$. By \Cref{lem:inducedPartitionProp} there are at least $\beta+1$ nodes $w' \in N(u)$ with $\sigma(w')=\infty$. Hence, it holds that $\sigma(w)=\infty$ for every node $w \in F(\sigma,u)$. Since $|N(u) \setminus D(\ell,u)| \leq \beta$, $|F(\sigma,u)|=\beta+1$, and $F(\sigma,u) \subseteq N(u)$, there must be at least one node $w \in F(\sigma,u) \cap D(\ell,u)$. Because $w \in F(\sigma,u)$, it holds that $\sigma(w) = \infty$. Because $w \in D(\ell,u) \setminus \{u\}$, it holds that $\ell(w) < \infty$, proving the claim.
	
	Let us consider the case where $\sigma(u) \neq \infty$. By \Cref{lem:inducedPartitionProp}, it holds that $\sigma(w)\geq \sigma(u)-1$ for at least $\beta+1$ nodes $w \in N(u)$. Hence, it holds that $\sigma(w)\geq \sigma(u)-1$ for every node $w \in F(\sigma,u)$. Since $|N(u) \setminus D(\ell,u)| \leq \beta$, $|F(\sigma,u)|=\beta+1$, and $F(\sigma,u) \subseteq N(u)$, there must be at least one node $w \in F(\sigma,u) \cap D(\ell,u)$. Because $w \in F(\sigma,u)$, it holds that $\sigma(w) \geq \sigma(u)-1 \geq \ell(u)$. Because $w \in D(\ell,u)$, it holds that $\ell(u) > \ell(w)$. We conclude that $\sigma(w) > \ell(w)$.
\end{proof}

The following two lemmas bring everything together to prove that after $x^2$ super-iterations of the coin dropping game it holds that $\sigma_{S_v,\beta}(v)=\ell_\beta(v)$ and $|S_v|\leq x^3+1$.

\begin{lemma}[Correctness] \label{lem:correctness}
	Consider a graph $G=(V,E)$, any $\beta,x\geq1$ and the natural $\beta$-partition $\ell_{\beta}$. Consider the $(x,\beta,F)$-coin dropping game for a fixed node $v \in G$ and the corresponding set $S_v \subseteq V$. If $|D(\ell_\beta,v)|\leq x^2$ and $\ell_\beta(v) \leq \log_{\beta+1} x$, then upon termination $\sigma_{S_v,\beta}(v)=\ell_\beta(v)$.
\end{lemma}

\begin{proof}
	Recall that the $(x,\beta,F)$-coin dropping game comprises of $x^2$ super-iterations. Consider $v \in V$ as in the lemma statement. Towards a contradiction, assume that $\sigma_{S_v,\beta}(v) \neq \ell_\beta(v)$ after $x^2$ super-iterations. By \Cref{lem:naturalBest}, this implies that $\sigma_{S_v,\beta}(v) > \ell_\beta(v)$ holds for $x^2$ super-iterations. Observe that $S_v$ can only increase in size during \Cref{alg:coinDropping}. Since $\sigma_{S_v,\beta}(v) > \ell_\beta(v)$ in every super-iteration, a node $w \in D(\ell_\beta,v) \setminus S_v$ is added to $S_v$ by \Cref{lem:measureOfprogress} in every super-iteration. Since $|D(\ell_\beta, v)|\leq x^2$, after $x^2$ super-iterations it will hold that $D(\ell_\beta,v) \subseteq S_v$, implying $\sigma_{S_v,\beta}(v) = \ell_\beta(v)$ by \Cref{lem:DGinS}, a contradiction.
\end{proof}

\begin{remark}[Nodes may not find their dependency graphs]
	In \Cref{lem:correctness} we do not guarantee that eventually $D(\ell_\beta,v) \subseteq S_v$, which would directly yield that $\sigma_{S_v,\beta}(v)=\ell_\beta(v)$ by \Cref{lem:DGinS}. Rather we prove only prove the sufficient statement $\sigma_{S_v,\beta}(v)=\ell_\beta(v)$ that can also happen even when $D(\ell_\beta,v) \not\subseteq S_v$.
\end{remark}

\begin{lemma}[Query Complexity] \label{lem:queryComplexity}
	Consider a graph $G=(V,E)$, any $\beta,x\geq1$ and the $(x,\beta,F)$-coin dropping game for a fixed node $v \in V$ and the corresponding set $S_v \subseteq V$. Upon termination, it holds that $G[S_v]$ is connected and contains at most $x^6$ edges. 
\end{lemma}

\begin{proof}
	During the game, a node can be added to $S_v$ only if it is adjacent to a node already in $S_v$. Hence, $G[S_v]$ is always connected. 
	
	Consider one super-iteration of the $(x,\beta,F)$-coin dropping game. Observe that the forwarding set of any node $u \in S_v$ is of size $\leq \beta+1$, and only nodes with $\geq \beta +1$ coins forward coins. Hence, if a node $w \not\in S_v$ receives coins, it receives $\geq 1$ coins. Since the total number of usable coins in a super-iteration is $x$, the number of nodes that can be added to $S_v$ in a super-iteration is at most $x$. The claim follows since there are $x^2$ super-iterations and $G[S_v]$ is connected.
\end{proof}

\subsection{The \lca} \label{sec:computePartitionARB}

In this section, we show the \lca for computing a partial $\beta$-partition with a large fraction of nodes being assigned a layer $ \neq \infty$. In particular, we prove the following lemma. 

\begin{lemma}[Sublinear \lca for partial $\beta$-partitions] 
	\label{lem:coinDropLCAformal} 
	For any constant $\eps>0$ and $\beta\geq (2+\eps)\alpha$ there is a deterministic \lca algorithm that uses at most $x^6$ queries per node on a graph $G$ with arboricity $\alpha$ and assigns each node a layer from $\mathbb{N}\cup \{\infty\}$ such that the following holds: 
	
	$\triangleright$ There exists a subset $S\subseteq V$ containing at least a $1-2^{1-\log x/\log_{\beta/(2\alpha)} (\beta+1)}$ fraction of vertices such that the layering of $S$ forms a $\beta$-partition of $G[S]$ with $\log_{\beta+1} x$ layers. $\triangleleft$
\end{lemma}

\begin{proof}
	Consider $\eps$ and $\beta$ as in the lemma statement. When questioned about vertex $v$, the \lca performs the $(x,\beta,F)$-coin dropping game and outputs the $S_v$-induced $\beta$-partition layer of $v$. In order to perform the coin dropping game, the \lca needs to be able to compute the $S_v$-induced $\beta$-partition (\Cref{def:inducedPartition}) locally. This is possible since it only requires knowing $G[S_v]$ and the degrees of nodes in $S_v$. By \Cref{lem:queryComplexity}, subgraph $G[S_v]$ contains at most $x^6$ edges at any point during the game, bounding the query complexity of the \lca. 
	
	By \Cref{lem:correctness}, if $|D(\ell_\beta,v)|\leq x^2$ and $\ell_\beta(v) \leq \log_{\beta+1} x$, then upon termination the \lca can output $\ell_\beta(v)$, which is the layer of $v$ in the natural $\beta$-partition of graph $G$. By \Cref{lem:removeEnoughNodes}, there are at least $|V|\cdot (1-2^{1-\log x/\log_{\beta/(2\alpha)} (\beta+1)})$ such nodes, completing the proof.
\end{proof}

\begin{remark} \label{rem:LCAwithMinimum}
	Instead of \Cref{lem:coinDropLCAformal}, we actually obtain a stronger statement. Our \lca, if queried about a node $u$, outputs a partial $\beta$-partition $\ell_u: V\rightarrow \mathbb{N}\cup\{\infty\}$.
	
	$\triangleright$ There exists a subset $S\subseteq V$ containing at least a $1-2^{1-\log x/\log_{\beta/(2\alpha)} (\beta+1)}$ fraction of vertices, and $\lambda: V\rightarrow \mathbb{N}\cup\{\infty\}: v \mapsto \min_{u\in S}\ell_u(v)$ is a partial $\beta$-partition of $G$ with $\log_{\beta+1} x$ layers. $\triangleleft$
\end{remark}

\begin{proof}
	For a node $u \in V$, instead of only outputting the $S_v$-induced $\beta$-partition layer $\ell_u(u)$ in the proof of \Cref{lem:coinDropLCAformal}, the \lca can simply output the whole partition $\ell_u$. By \Cref{lem:closureUnderMinimum}, function $\lambda$ is a partial $\beta$-partition of $G$ with the desired properties.
\end{proof}

The following lemmas are purely combinatorial and do not depend on the coin dropping game nor any model of computation.

\begin{lemma} \label{lem:removeEnoughNodes}
	Consider a graph $G=(V,E)$ with arboricity $\alpha$ and the natural $\beta$-partition $\ell_{\beta}$ for any constant $\eps>0$ and $\beta\geq (2+\eps)\alpha$. For any $x>1$, let $X \subseteq V$ such that $|D(\ell_\beta,v)| \leq x^2$ and $\ell_{\beta}(v)\leq \log_{\beta+1} x$ for $v \in X$. It holds that 
	\begin{equation*}
		|V \setminus X| \leq |V| \cdot 2^{1-\log x/\log_{\beta/(2\alpha)} (\beta+1)}~.
	\end{equation*}
\end{lemma}

\begin{proof}
	Divide the nodes of $V$ into sets $V_{>}$ and $V_{\leq}$, such that $\ell_\beta(w)> \log_{\beta+1} x$ holds for nodes $w \in V_{>}$ and $\ell_\beta(w)\leq \log_{\beta+1} x$ holds for nodes $w \in V_{\leq}$. Furthermore, let us define $V_{\leq}^X \coloneqq V_{\leq} \cap X $. Observe that $V \setminus X = V_{>} \cup V_{\leq}^X$. Let us bound the terms separately. 
	
	\begin{itemize}
		\item[$V_{\leq}^X$:] By definition, it holds that $|D(\ell_\beta,w)|>x^2$ for nodes $w \in V_{\leq}^X$. Towards a contradiction, assume that $|V_{\leq}^X| > |V| \cdot x^{-1}$. Let every node $w \in V_{\leq}^X$ give every node in its dependency graph one coin. Note that all nodes in the dependency graph of $w$ are in $V_{\leq}$. The total number of coins is hence at least $|V| \cdot x^{-1} \cdot x^2= |V| \cdot x$. On the other hand, any node $u$ with $\ell_\beta(u) \leq \log_{\beta+1} x$ has at most $\beta$ neighbors in higher layers, and hence can be given a coin by at most $\beta^{\log_{\beta+1} x} < x$ nodes in $V$. Since the total number of coins is at least $|V_{\leq}| \cdot x$, and every node can only be charged at most $x$ coins, it must be that
		$
		|V|> (|V| \cdot x)/x = |V|~,
		$
		which is a contradiction and hence it holds that $|V_{\leq}^X| \leq |V| \cdot x^{-1}$.
		
		\item[$V_{>}$:] By \Cref{lem:generalManySmallDegreeNodes}, we get that for each layer $i$ the number of nodes that will receive a layer larger than $i$ shrinks by a factor of $2\alpha/\beta$. Therefore, the number of nodes with $\ell_\beta(w)> \log_{\beta+1} x$ is 
		\begin{equation*}
			|V_{>}| \leq |V| \cdot (2\alpha/\beta)^{\log_{\beta+1} x} = |V| \cdot (\beta/(2\alpha))^{-\log_{\beta+1} x} = |V| \cdot x^{-1/\log_{\beta/(2\alpha)} (\beta+1)}~.
		\end{equation*}
	\end{itemize}
	Hence, it holds that $|V \setminus X| \leq |V| \cdot x^{-1} + |V| \cdot x^{-1/\log_{\beta/(2\alpha)}(\beta+1)} \leq |V| \cdot 2^{1-\log x/\log_{\beta/(2\alpha)} (\beta+1)}$.
\end{proof}

\begin{lemma}[Closure under minimum for partial $\beta$-partitions] \label{lem:closureUnderMinimum}
	Let $\ell_1,\dots,\ell_k$ be partial $\beta$-partitions of a graph $G=(V,E)$ for any $\beta\geq 1$. The following function $\lambda: V\rightarrow \mathbb{N} \cup \{\infty\}$ assigning a layer to every node $v\in V$ is also a partial $\beta$-partition:
	\begin{align*}
		\lambda(v) = \min_{1\leq i\leq k} \ell_i(v) ~.
	\end{align*}
	
	Moreover, if $\ell_i(v) \neq \infty$ for any $i \in [1,k]$ then $\lambda(v) \neq \infty$.
\end{lemma}

\begin{proof}
	In order to prove that $\lambda$ is a partial $\beta$-partition, we need to show that $\lambda(w)\geq \lambda(v) \neq \infty$ for at most $\beta$ nodes $w \in N(v)$. Let $N^{\geq}(v)$ be the set of neighbors $w \in N(v)$ with $\lambda(w)\geq \lambda(v)$.
	
	Without loss of generality, we can assume that $\lambda(v)=\ell_1(v)$. Let $N^{<}_1(v)$ be the set of neighbors $w \in N(v)$ with $\ell_1(w)<\ell_1(v)$, and let $N^{\geq}_1(v)$ be the set of neighbors $w \in N(v)$ with $\ell_1(w)\geq \ell_1(v)$. Let us count how many nodes from sets $N^{<}_1(v)$ and $N^{\geq}_1(v)$ could be in $N^{\geq}(v)$.
	
	\begin{itemize}
		\item Because $\lambda(w)=\min\{\ell_1(w),\dots,\ell_k(w)\}$ for every node $w\in N^{<}_1(v)$, it holds that $\lambda(w)\leq \ell_1(w)$. Hence, no nodes from $N^{<}_1(v)$ can be in $N^{\geq}(v)$.
		\item Because $\ell_1$ is a partial $H$-partition, it holds that $|N^{\geq}_1(v)| \leq \beta$ so at most $\beta$ nodes from $N^{\geq}_1(v)$ can be in $N^{\geq}(v)$.
	\end{itemize}
	
	We conclude that $|N^{\geq}(v)|\leq \beta$, proving that $\lambda$ is a partial $\beta$-partition. If $\ell_i(v) \neq \infty$ for any $i \in [1,k]$ then $\lambda(v) \neq \infty$ because function $\lambda$ takes the minimum and $\ell_i(v) < \infty$.
\end{proof}

\section{$\beta$-partition in \ampc} \label{sec:AMPCPartition}

This section is dedicated to constructing a $\beta$-partition from \Cref{def:hPartition} in the \ampc model. Observe that we mean to construct a $\beta$-partition with an empty infinity layer. In particular, we prove \Cref{thm:partitioning}, which assumes the knowledge of the arboricity. We later remove this assumption in \Cref{lem:withoutArb}.

\partitioning*

\begin{proof}
	Consider an graph $G=(V,E)$ with $n$ nodes and arboricity $\alpha< n^{\frac{(\delta/c)^2 \cdot \log ((2+\eps)/2)}{2(2+\eps)}}$ (this requirement is lifted later in the proof). The algorithm operates on data stores $D_0,D_1,\dots$ such that the total number of data stores equals the number of rounds of the algorithm.
	Let us keep track of two types of key-value pairs in every data store $D_i$. The first type being a value from $\mathbb{N} \cup \{\infty \}$ returned for every node $v \in V$, which represents $v$'s layer in a (partial) $\beta$-partition. The second type is defined using the first type as follows. Let $G_i=(V_i,E_i)$ be the subgraph of $G=(V,E)$, induced by nodes $v$ with $D_i(v) = \infty$, which has arboricity at most $\alpha$ by \Cref{fact:ArbDec}. The edges of graph $G_i$ are stored in the data store $D_i$ as key-value pairs $(v,\deg_{G_i}(v),j) \arr (u,\deg_{G_i}(u))$, where $j\in [0,\deg_{G_i}(v)-1]$ and node $u$ is the $j$th neighbor of $v$ in $G_i$, according to some arbitrary ordering. For $D_0$, the values of the first type are set to $\infty$, and hence $G_0=G$ is just the input graph. 
	
	Let $P=n$ be the number of machines\footnote{\cite{Behnezhad2019-remote} provides the necessary details for employing this number of machines using \emph{parallel slackness}.} and $S = \Omega(n^\delta)$ be the local space of a machine, for any $\delta > 0$. In order to compute $D_{i+1}$ from $D_i$, we have the following protocol. Let $M_v$ be the machine responsible for node $v \in V$. Using $D_i$, or more precisely graph $G_i$, machine $M_v$ performs the \lca of \Cref{rem:LCAwithMinimum}. The query complexity of the \lca is $x^6$, so when $x=n^{\delta/c}$ for some $c>6$, the local and total space of our \ampc algorithm are respected. 
	
	Next, we need to argue that the partial $\beta$-partition of $G_i$ from \Cref{rem:LCAwithMinimum} is computable.
	After performing the \lca from \Cref{rem:LCAwithMinimum}, every machine $M_v$ has the $S_v$-induced $\beta$-partition $\sigma_{S_v,\beta}$ stored in memory. Machine $M_v$ can write the key-value pair $u \arr \sigma_{S_v,\beta}(u)$ for every node $u \in S_v$ for which $\sigma_{S_v,\beta}(u) \leq \log_{\beta+1} x$ to data store $D_{i+1}$. The separate set of machines that handles the DDS sorts through all written values and keeps the smallest $\sigma(u)$ value for every node $u \in V_i$. Observe that in order to make the (final outputted) $\beta$-partition consistent throughout the algorithm, we need to record $\sigma(u)$ plus the maximum numbered layer of any node in $D_i$, effectively appending the layers. Additionally, using $D_i$, these machines can compute $\deg_{G_{i+1}}(u)$ for every node $u \in G_{i+1}$ and port the edges of $G_{i+1}$ to $D_{i+1}$. 
	
	Let us analyze the size of $|V_i|$. By applying \Cref{rem:LCAwithMinimum} for every rounds, it holds that
	\begin{align}
		|V_{i}| &\leq n \cdot 2^{i \cdot (1 - \log x/\log_{\beta/(2\alpha)} (\beta+1))} \nonumber\\ 
		&= 2^{i/ (\log_{\beta/(2\alpha)} (\beta+1)) \cdot (\log_{\beta/(2\alpha)} (\beta+1) - \log x)} \label{line:drop}~. 
	\end{align}
	
	When $x=n^{\delta/c}$, it holds that 
	\begin{align}
		\log_{\beta/(2\alpha)} (\beta+1) - \delta/c \cdot \log n 
		&< \frac{\log 2\beta}{\log \beta - \log 2\alpha} - \delta/c \cdot \log n \nonumber\\
		&\stackrel{*}{\leq} \frac{\log (2(2+\eps)\alpha)}{\log ((2+\eps)\alpha) - \log 2\alpha} - \delta/c \cdot \log n \nonumber\\
		&= \log_{(2+\eps)/2} (2(2+\eps)\alpha)) - \delta/c \cdot \log n \nonumber\\
		&< \log_{(2+\eps)/2} (\alpha^{2(2+\eps)})) - \delta/c \cdot \log n \nonumber\\
		&< (\delta/c)^2 \cdot \log ((2+\eps)/2) \cdot \frac{\log n}{\log ((2+\eps)/2)} - \delta/c \cdot \log n \nonumber\\
		&= -(\delta/c - (\delta/c)^2) \cdot \log n~, \label{line:insert}
	\end{align}
	
	where in $*$ we use $\beta \geq (2+\eps)\alpha$. Let $d \coloneqq (\delta/c - (\delta/c)^2)$. Observe that since $\delta/c<1$, it holds that $d > 0$. Plugging \Cref{line:insert} into \Cref{line:drop} gives
	\begin{align*}
		|V_{i}| &\leq 2^{i/ (\log_{\beta/(2\alpha)} (\beta+1)) \cdot (-d \log n)}~. 
	\end{align*}
	
	Hence, after $k=\log_{\beta/(2\alpha)} (\beta+1)/d = O(\log_{\beta/(2\alpha)} \beta)$ rounds, it holds that 
	\begin{align*}
		|V_{k}| \leq n \cdot 2^{-\log n} = 0~.
	\end{align*}
	
	By \Cref{rem:LCAwithMinimum}, the maximum layer in the partial $\beta$-partition computed for every data store is $\log_{\beta+1} x$. Hence, after $O(\log_{\beta/(2\alpha)} \beta)$ rounds, the size of the $\beta$-partition is $=O(\log_{\beta/(2\alpha)} n)$.
	
	When $\alpha \geq n^{\frac{(\delta/c)^2 \cdot \log ((2+\eps)/2)}{2(2+\eps)}}$, we employ a different algorithm altogether. We simulate the $H$-partition algorithm of \cite{BE10}. For every data store $D_i$, instead of performing the $(x,\beta,F)$-coin dropping game, every node of degree $\leq \beta$ simply places itself in layer $i$. By \Cref{lem:generalManySmallDegreeNodes}, for any $\beta \geq (2+\eps)\alpha$, this algorithm results in a $\beta$-partition of size $O(\log_{\beta/(2\alpha)} n)$ in $O(\log_{\beta/(2\alpha)} n) = O(\log_{\beta/(2\alpha)} \beta)$ time.
\end{proof}

\Cref{thm:partitioning} assumes the knowledge of the arboricity $\alpha$ of the graph. In \Cref{lem:withoutArb} we remove this assumption.

\begin{lemma} \label{lem:withoutArb}
	\Cref{thm:partitioning} can be applied without the knowledge of the arboricity of the graph. 
\end{lemma}

\begin{proof}[Proof of \Cref{lem:withoutArb}]
	Let $\beta$ be defined using arboricity $\alpha$. When $\alpha$ is unknown, we can still find the $\beta$-partition with no additional cost in the runtime or memory by employing a guessing scheme for $\alpha$.
	
	First, we perform a series of sequential executions of \Cref{thm:partitioning}, such that instance $i$ computes a $\beta_i$-partition, using the $i$th guess for arboricity $\alpha_i=2^{2^i}$, $i \in [0,\log \log n]$. We terminate after the first guess $\alpha_k$ which returns a $\beta_k$-partition. Due to the double exponential guessing and the fact that a correct guess we yield partition, we know that $a_k < \alpha^2$.
	
	Next, we perform a number of parallel executions of \Cref{thm:partitioning}, such that instance $i$ computes a $\beta_i$-partition, where the $i$th guess for arboricity is $\alpha_i=\sqrt{a_k} \cdot (1+\eps)^i$, $i \in [0,\log_{1+\eps} \sqrt{a_k}]$ for some $\eps>0$. Observe that for some instance $j$, it must hold that $\alpha_j < (1+\eps)\alpha$, producing a $\beta_i$-partition of size $O(\log_{\beta_i/2\alpha_i} n)=O(\log_{\beta/2\alpha} n)$. As our solution, we pick the smallest instance that produces a partition.
	
	Let us analyze the total runtime when employing the guessing scheme above. The sequential executions take
	\begin{align*}
		\sum_{i=0}^{k} 
		O(\log_{\beta_i/2\alpha_i} \beta_i) \leq \sum_{i=0}^{\infty} \left( \frac{1}{2^i}\right) O(\log_{\beta/2\alpha} \alpha) + O(\log_{\beta/2\alpha} \alpha^2) = O(\log_{\beta/2\alpha} \alpha) = O(\log_{\beta/2\alpha} \beta)
	\end{align*}
	
	rounds in total. The parallel executions take $O(\log_{\beta/2\alpha} \alpha^2) = O(\log_{\beta/2\alpha} \beta)$ rounds. The parallel executions require an additional $\log_{1+\eps} \sqrt{a_k} = O(\log n)$ factor of space, which can be allowed by scaling the constant $\delta$, corresponding to the local space $O(n^\delta)$.
	
	It is worth mentioning that since $\alpha_j$ may differ from $\alpha$ by a $(1+\eps)$ factor, the resulting $\beta_j$ may also differ from $\beta$ by some $\eps'>1$ factor. This can be mended by simply using a slightly smaller $\beta$ value to begin with, which is always possible due to the slack in requirement $\beta \geq (2+\epsilon)\alpha,~\epsilon>0$. 
\end{proof}
\section{Coloring Graphs in \ampc}
\label{sec:coloring}
In this section we will prove the three results in \Cref{thm:coloring}. We first restate the theorem as a reference for the reader.

\coloringthm*
We will prove the first result in \Cref{sec:fastColoring}, the second result it \Cref{sec:mediumColoring}, and the third result in \Cref{sec:slowColoring}. In all three cases we proceed by simulating a \local algorithm in the \ampc model by showing that in order to simulate $x$ rounds of the \local algorithm, each node needs to know at most $O(n^\delta)$ sized information around its neighborhood. This simulation can be done in one round of \ampc by assigning each node to a dedicated machine which will adaptively gather the neighborhood information and locally simulate $x$ rounds of the \local algorithm. In order for the \local algorithm simulation to work, we will assume that $\alpha \le n^{\delta/(1+\eps)}$ for some constant $\eps>0$. We handle case of $\alpha > n^{\delta/(1+\eps)}$, separately in \Cref{sec:handlingLargeAlpha}. 

\subsection{An $O(\alpha^{2+\eps})$-coloring in $O(1/\eps)$ Rounds}
\label{sec:fastColoring}
Compute a $\beta$-partition using \Cref{thm:partitioning} with $\beta = O(\alpha^{1+\eps})$ for a constant $\eps>0$. This will require $O(1/\eps)$ \ampc rounds and the $\beta$-partition will have $O(\eps^{-1}\log_{\alpha}n)$ layers. 

To compute this coloring, we simulate the Arb-Linial coloring algorithm of \cite{BE10} at each node. The Arb-Linial algorithm takes as input a $\beta$-partition and produces an $O(\beta^2)$-coloring which using $O(\log^* n)$ rounds in the \local model. Let us orient all edges from lower layer to higher layer in the $\beta$-partition and edges between nodes in the same layer are oriented arbitrarily. An important property of the Arb-Linial algorithm is that it is \emph{one-sided}, that is a node $v$ only needs to know its out-neighborhood to compute its output. 

In the algorithm we start with a $\poly(n)$-coloring where each node picks its ID as its color. In one round of algorithm, an $m$-coloring is reduced to an $O(\beta^2\log m)$-coloring. Therefore, after one round, we have an $O(\beta^2\log n)$-coloring, and this eventually results in an $O(\beta^2\log \beta)$-coloring in $O(\log^* n)$ rounds. Then in one final round, the $O(\beta^2\log \beta)$-coloring is reduced to an $O(\beta^2)$-coloring.

Note that if $\alpha^{\eps} > \log n$, we don't need to simulate all $\log^* n$ rounds of the Arb-Linial algorithm. If we just simulate the first round of Arb-Linial, we get an $O(\beta^2\log n)$-coloring. Since $\beta = O(\alpha^{1+\eps})$ and $\log n < \alpha^{\eps}$, $O(\beta^2\log n) = O(\alpha^{2+3\eps})$, which we can say is $O(\alpha^{2+\eps})$, by replacing $\eps$ by $\eps/3$. We can simulate the first round of Arb-Linial in $1$ \ampc round by gathering the at most $\beta$ out-neighbors of each node. We can do this since we assumed $\alpha \le n^{\delta/(1+\eps)}$, which implies $\beta \le n^{\delta}$.

Now we describe how to proceed when $\alpha^{\eps} \le \log n$. In this case, we simulate all $O(\log^* n)$ rounds of the Arb-Linial algorithm in $1$ \ampc round by collecting the $\beta^{O(\log^* n)}$ size out-neighborhood at each node $v$. In order for this to work, we need the entire $O(\log^*n)$-hop out-neighborhood to fit into the memory of a single machine, that is, $\beta^{O(\log^* n)} \le n^\delta$ which is true since $\alpha^{\eps} \le \log n$. Therefore, we get number of colors to be $O(\beta^2) = O(\alpha^{2+2\eps})$, which we can say is $O(\alpha^{2+\eps})$, by replacing $\eps$ by $\eps/2$.

\subsection{An $O(\alpha^2)$-coloring in $O(\log \alpha)$ Rounds}
\label{sec:mediumColoring}
Compute a $\beta$-partition using \Cref{thm:partitioning} with $\beta = (2+\eps)\alpha$ for a constant $\eps>0$. This will require $O(\eps^{-1}\log \alpha)$ \ampc rounds and the $\beta$-partition will have $O(\eps^{-1}\log n)$ layers.

If $\alpha > 2^{\log^* n}$, we simulate the Arb-Linial algorithm one round at a time. We can simulate one round of Arb-Linial in $1$ \ampc round by gathering the at most $\beta$ out-neighbors of each node. We can do this since we assumed $\alpha \le n^{\delta/(1+\eps)}$, which implies $\beta \le n^{\delta}$. The overall simulation will require $O(\log^* n) = O(\log \alpha)$ \ampc rounds.

If $\alpha \le 2^{\log^* n}$, we simulate all $O(\log^* n)$ rounds of the Arb-Linial algorithm in $1$ \ampc round by collecting the $\beta^{O(\log^* n)}$ size out-neighborhood at each node $v$. In order for this to work, we need $\beta^{O(\log^* n)} \le n^\delta$ which is true since $\alpha \le 2^{\log^* n}$. In both cases, we get number of colors to be $O(\beta^2) = O(\alpha^2)$.

\subsection{A $((2+\eps)\alpha + 1)$-coloring in $O(\alpha\log \alpha)$ rounds} 
\label{sec:slowColoring}
Compute a $\beta$-partition using \Cref{thm:partitioning} with $\beta = (2+\eps)\alpha$ for a constant $\eps>0$. This will require $O(\eps^{-1}\log \alpha)$ \ampc rounds and the $\beta$-partition will have $O(\eps^{-1}\log n)$ layers.

We first construct an initial (improper) coloring of the nodes using $(\beta+1)$-colors. To do this we use a $(\Delta+1)$-coloring algorithm by Kuhn and Wattenhofer \cite{Kuhn2006On} which runs in $O(\Delta \log \Delta + \log^* n)$ rounds in the \local model. This algorithm works in two stages, in the first stage they compute an $O(\Delta^2)$-coloring in $O(\log^* n)$ rounds using Linial's coloring algorithm \cite{linial92}, and in the second stage, they apply an iterative color reduction technique that reduces the $O(\Delta^2)$-coloring to a $(\Delta + 1)$-coloring in $O(\Delta \log \Delta)$ rounds. We will use this algorithm to find an initial coloring of the nodes. We first compute a $(\beta + 1)$-coloring of the subgraph induced by nodes in a single layer by simulating the Kuhn-Wattenhofer algorithm independently in each layer in parallel.

If $\alpha > \log^* n$, we can simulate the entire algorithm one round at a time in $O(\alpha \log \alpha)$ \ampc rounds. Again, we can simulate a single round of the algorithm in $1$ \ampc round by gathering the at most $\beta$ out-neighbors of each node. We can do this since we assumed $\alpha \le n^{\delta/(1+\eps)}$, which implies $\beta \le n^{\delta}$.

If $\alpha \le \log^* n$, we simulate the two stages separately. First, we compute the $O(\beta^2)$-coloring in $O(1)$ \ampc rounds by having each node collect a $O(\log^* n)$-radius ball around it. In order for the simulation to work, we need that $\alpha^{O(\log^* n)} \le n^{\delta}$ which is true for $\alpha \le \log^* n$. Second, we simulate the color reduction stage one round at a time, which requires $O(\beta \log \beta)$ \ampc rounds. The initial coloring is not yet a proper coloring of the entire graph, since we can have conflicts at edges between nodes in two different layers. 

We fix these conflicts greedily by simulating the following centralized recoloring process: all nodes in the topmost layer fix their initial color to be their final color. Start with the nodes in the second highest layer and go down to the lowest layer, one layer at a time. In each layer $i$, process nodes in decreasing order of their initial color and each node picks the highest available color such that there are no conflicts with neighbors that have finalized their color. 

Note that only nodes in the same or higher layer can finalize their color, and a node $v$ has at most $\beta$ neighbors in the same or higher layer. So at most $\beta$ colors in their palette of size $\beta+1$ are blocked, so there will always be at least one color $v$ can pick so that there are no conflicts with neighbors that have finalized their color.

In order to simulate the recoloring process in \ampc, we make the observation that the final color of a node $v$ only depends on the final colors of neighbors in the same layer having higher initial color, and the neighbors in higher layers. In other words, if we want to compute the final color of node $v$, we first need to compute the final colors of all these neighbors, so that we can find the highest available color for $v$. We compute the final colors of the neighbors of $v$ recursively, i.e., by computing the final colors of their neighbors in the same layer having higher initial color, and the neighbors in higher layers. Therefore, in order to compute the final color of a node $v$, we will need to compute the final color of all nodes $u$ such that the each edge in the path from $v$ to $u$ either goes from one initial color to a higher initial color or goes from one layer to a higher layer.

Since there are $O(\eps^{-1}\log n)$ layers and $(\beta+1)$ initial colors, all such nodes $u$ are at most $O((\beta/\eps) \log n)$ hops away from $v$. Moreover, all such nodes $u$ are in the same or higher layer as $v$, which means that to compute the final color of $v$ we need to collect a subgraph of size at most $\beta^{O((\beta/\eps) \log n)} = n^{O((\beta/\eps) \log \beta)})$. This is too large to fit in the memory of a single machine, so we process the layers in batches of size $(c\delta/\beta) \log_\beta n$, where $c$ is a constant chosen such that a node in this batch of layers only has to collect a ball of size $\beta^{O(\beta \cdot (c\delta/\beta)\log_{\beta} n)} = O(n^\delta)$, and there are $O((\beta/\eps \delta)\log \beta)$ batches. Each batch requires $1$ \ampc round so we use $O((\beta/\eps \delta)\log \beta) = O(\alpha \log \alpha)$ \ampc rounds for the overall simulation of the recoloring process.

\subsection{Coloring when $\alpha$ is Large}
\label{sec:handlingLargeAlpha}

We begin with proving \Cref{thm:derandomizedColoring}, which we will use as an important building block. For ease of reading, we restate the theorem before providing the proof.

\derandomizedcoloring*
\begin{proof}
	We will prove this theorem by providing a deterministic, non-component-stable, \mpc algorithm that produces the coloring in $O(\log_x n)$ rounds using $O(n+m)\poly\log n$ global memory. Since the \ampc model can simulate any \mpc algorithm without any loss in round complexity, the theorem follows.
	
	Consider the following randomized trial with pairwise independent random coins: Every node picks a random color from $1, \dots, 2x\Delta$. Let $Y_e$ be an indicator random variable for the event that edge $e$ is monochromatic, and let $Y = \sum_e Y_e$ be the number of monochromatic edges produced by the randomized trial. Since we use pairwise independent random coins, the probability that $Y_e = 1$ is $1/(2x\Delta)$. By the handshaking lemma and linearity of expectation, $\E[Y] \le \sum_{v} \sum_{e \ni v} \E[Y_e] = \sum_{v} \deg(v)/(2x\Delta) \le |V|/(2x)$.
	
	We derandomize this process using the method of conditional expectations, which was first introduced in the context of the \mpc model in \cite{CPS20}. Since we use pairwise independence, we can reduce the number of random coins required to $O(\log n)$-bits by using a family of pairwise independent hash functions. Then we fix these random bits in batches of size $\lfloor (\delta/3) \log_2 n \rfloor$, and we denote the $i^{th}$ batch by $B_i$. The invariant we will maintain when fixing the random bits in each batch is that $\E[Y | B_1 = b_1, \dots, B_i = b_i] \le \E[Y]$.
	
	When fixing batch $B_i$, let us assume inductively that the invariant is maintained for all previous batches, i.e. we have fixed the random bits such that $\E[Y | B_1 = b_1, \dots, B_{i-1} = b_{i-1}] \le \E[Y]$. By the law of total expectations, there exists an assignment $B_i = b_i$ such that the invariant is maintained. There are at most $n^{\delta/3}$ assignments to $B_i$, and we will calculate the value of the conditional expectation $\E[Y | B_1 = b_1, \dots, B_i = b_i]$ for each assignment $b_i$ in parallel and pick the minimum.
	
	Let us see first how to calculate the conditional expectation for a single assignment $B_i= b_i$. The key idea is that we decompose $\E[Y | B_1 = b_1, \dots, B_i = b_i]$ into $\sum_{e}\E[Y_e | B_1 = b_1, \dots, B_i = b_i]$, so the task is to compute a sum of values where each value can be computed locally by the machine that stores edge $e$. We create a ``broadcast tree'' on the set of machines, which is a balanced $n^{\delta/2}$-ary tree of depth $O(1/\delta)$. Each machine $M$ locally computes $\E[Y_e | B_1 = b_1, \dots, B_i = b_i]$ for all edges $e$ stored at $M$ and adds them up. The leaves in the broadcast tree simply send this sum to their parent. Once a machine in the broadcast tree receives the values from all its children, it sends the sum of all these values and its locally computed sum to its parent. In each round, a machine sends at most one value and receives at most $n^{\delta/2}$ values. In $O(1/\delta)$ rounds, the root of the broadcast tree will calculate the value of the conditional expectation $\E[Y | B_1 = b_1, \dots, B_i = b_i]$.
	
	Since there are at most $n^{\delta/3}$ assignments to $B_i$, it is easy to see that the machines can calculate the conditional expectation value for all assignments $b_i$ in parallel over the same broadcast tree. In each round, a machine sends at most $n^{\delta/3}$ values and receives at most $n^{5\delta/6}$ values ($n^{\delta/3}$ values from each of its $n^{\delta/2}$ children in the broadcast tree). Therefore the root of the tree can pick the assignment $b_i$ which minimizes the conditional expectation, and this maintains the invariant for our assignment to $B_i$. Since there are $O(1/\delta)$ batches, we are able to fix all the random bits in $O(1/\delta^2)$ rounds.
	
	Now we have a deterministic algorithm that colors nodes with colors $1, \dots, 2x\Delta$, such that there are at most $n/(2x)$ monochromatic edges. We fix the color of all nodes that do not have an incident monochromatic edge and all other nodes are uncolored. Let $U$ be the set of uncolored nodes. We have that $|U| \le n/x$. If we repeat the same randomized process for the uncolored nodes, i.e., each node in $U$ picks a color in $1, \dots, 2x\Delta$ using pairwise independent random bits. We still have that the expected number of monochromatic edges is $\E[Y] \le |U|/(2x) \le n/(2x^2)$. We again derandomize this process, and we get a coloring of nodes with colors $1, \dots, 2x\Delta$, such that there are at most $n/(2x^2)$ monochromatic edges. The new set of uncolored nodes has size at most $n/x^2$. Therefore, if we repeat the process $i$ times we get a partial coloring of nodes with colors $1, \dots, 2x\Delta$, such that the number of uncolored nodes is at most $n/x^i$. We continue until no uncolored vertices remain. The number of repetitions is at most $O(\log_x n)$, and we use at most $2x\Delta$ colors.
\end{proof}

We prove the first and second parts of \Cref{thm:coloring} by computing an $\alpha^{1+\eps}$ coloring in $O(1/\eps)$ rounds. Compute a $\beta$-partition using \Cref{thm:partitioning} with $\beta = \alpha^{1+\eps}$ for a constant $\eps>0$. This will require $O(1/\eps)$ \ampc rounds and the $\beta$-partition will have $O(\eps^{-1}\log_{\alpha}n) = O(1/\eps)$ layers.

We color each layer of the $\beta$-partition independently and in parallel with each layer getting a fresh palette using \Cref{thm:derandomizedColoring} with $x = \alpha^\eps$. Note that to run the algorithm, we just need to know for each edge what the layer numbers of the two end points are: if both end points belong to layer $i$, we include this edge in the algorithm for layer $i$, otherwise we ignore the edge, since it will never be monochromatic. Since the subgraph induced by each layer has maximum degree $\beta$, The total number of colors we use in a single layer is $O(\alpha^{\eps}\beta)$. Therefore, the total number of colors used over all layers is $O(\eps^{-1} \alpha^{\eps}\beta) \le \alpha^{1+3\eps}$. We can now adjust $\eps$ to be $\eps/3$ to obtain an $\alpha^{1+\eps}$ coloring. The coloring using \Cref{thm:derandomizedColoring} requires $O(1/\eps)$ \ampc rounds, so overall we require $O(1/\eps)$ \ampc rounds.

We now prove the third part of \Cref{thm:coloring}. First, we compute a $\beta$-partition using \Cref{thm:partitioning} with $\beta = (2+\eps)\alpha$ for a constant $\eps>0$. This will require $O(\eps^{-1}\log \alpha)$ \ampc rounds and the $\beta$-partition will have $O(\eps^{-1}\log n)$ layers. Then we construct an initial coloring where we color the subgraphs induced by each layer independently and in parallel using \Cref{thm:derandomizedColoring} with with $x = 2$. Since the max degree of the subgraph induced by nodes in a single layer is at most $\beta$, this creates a coloring of each layer with colors in $\{1, \dots, 4\beta\}$. Note that to run the algorithm, we just need to know for each edge what the layer numbers of the two end points are: if both end points belong to layer $i$, we include this edge in the algorithm for layer $i$, otherwise we ignore the edge, even though it can be monochromatic.

We orient all edges from lower layer to higher layer, and for edges with end points in the same layer, we orient them from lower initial color to higher initial color. Since we started with a $\beta$-partition, each node has at most $\beta$ outgoing neighbors. We use this orientation to produce a proper $(\beta + 1)$-coloring. We iterate through the layers sequentially from the topmost layer to the bottom most layer. Nodes in layer $i$ are processed in decreasing order of their initial color. Each node $v$ in layer $i$ with initial color $c$ (in parallel) picks as its final color the smallest available color in $\{1, \dots, \beta+1\}$. Due to the order in which we process the nodes, only the outgoing neighbors of $v$ have picked their final color before $v$ picks its final color. Therefore, at most $\beta$ colors are blocked for $v$, so it can always find an available color in $\{1, \dots, \beta+1\}$. Since $\beta$ can be much larger than the local space, we need to be a bit careful in finding the smallest available color. In order to do this, we sort in $O(1)$ rounds, the set of outgoing neighbors $u$ of nodes $v$ having layer $i$ and initial color $c$ with the key $(ID(v), \mathrm{col}(u))$, where $\mathrm{col}(u)$ is the final color picked by $u$ (constant round deterministic sorting is a well known \ampc/\mpc primitive \cite{Czumaj2020, Goodrich99, Goodrich2011}). This means that all outgoing neighbors of $v$ are stored in contiguous machines in increasing order of their final color. For the vertices whose set of outgoing neighbors spans a single machine, the smallest available color can be locally computed by that machine. For a vertex $v$ whose set of outgoing edges span multiple machines, we create a virtual broadcast tree (which is an $n^{\delta/2}$-ary balanced tree) on the set of machines that store the outgoing edges of $v$, and aggregate the smallest available color at the root. Note that a machine can be part of at most two such broadcast trees. It is easy to see that all the aggregations on the virtual broadcast trees can be implemented on a single broadcast tree on the set of all machines. This allows us to find in $O(1)$ rounds the smallest available color in $\{1, \dots, \beta+1\}$ for each node $v$ in layer $i$ with initial color $c$. Therefore, the time required for picking the final colors is $O(\beta \log n) = O(\alpha \log \alpha)$ rounds, which dominates the overall runtime.

\newpage
\bibliographystyle{alpha}
\bibliography{h-partition}
	
\end{document}